\RequirePackage{fix-cm}
\documentclass{article}
\usepackage{amsmath}
\usepackage{geometry}
\usepackage{graphicx}
\usepackage{amsfonts}
\usepackage{amsthm}
\usepackage{comment}
\usepackage[round]{natbib}
\usepackage{tikz}
\usepackage[font=normal]{subcaption}
\usepackage{hyperref}
\usepackage{xspace}
\newtheorem{theorem}{Theorem}
\newtheorem{lemma}{Lemma}
\newtheorem{definition}{Definition}
\newtheorem{proposition}{Proposition}
\usepackage{graphicx}
\allowdisplaybreaks
\newcommand{\Cactive}{C_{\mathit{active}}}
\newcommand{\Cinactive}{C_{\mathit{inactive}}}
\newcommand{\Lit}[1]{\mathit{Lit}[#1]}
\newcommand{\Cl}[2]{\mathit{Cl}[#1,#2]}

\newcommand{\idxtolit}[1]{\operatorname{l}(#1)}
\newcommand{\PQ}{\emph{long--short}\xspace}
\newcommand{\QP}{\emph{short--long}\xspace}

\begin{document}

\title{Scheduling with two non-unit job lengths is NP-complete
}


\author{Jan Elffers, 
              KTH Royal Institute of Technology, Stockholm, Sweden \\
              elffers@kth.se           
           \and
           Mathijs de Weerdt, Delft University of Technology, Delft, the Netherlands \\
              m.m.deweerdt@tudelft.nl
}

\maketitle

\begin{abstract}
The non-preemptive job scheduling problem with release times and deadlines on a single machine is fundamental to many scheduling problems.
We parameterize this problem by the set of job lengths the jobs can have.
The case where all job lengths are identical is known to be solvable in polynomial time.
We prove that the problem with two job lengths is NP-complete,
except for the case in which the short jobs have unit job length, which was already known to be efficiently solvable.
The proof uses a reduction from satisfiability to an auxiliary scheduling problem that includes a set of paired jobs that each have both an early and a late deadline, and of which at least one should be scheduled before the early deadline.
This reduction is enabled by not only these pairwise dependencies between jobs, but also by dependencies introduced by specifically constructed sets of jobs which have deadlines close to each other.
The auxiliary scheduling problem in its turn can be reduced to the scheduling problem with two job lengths by representing each pair of jobs with two deadlines by four different jobs.
\end{abstract}

\section{Introduction}
The problem considered in this paper is the non-preemptive job scheduling problem with release times and deadlines.
In the three-field notation, this problem is denoted as $1|r_i|L_{\max}$.
In this offline scheduling problem, there is a set of jobs, each having a release time, a deadline and a processing time,
that need to be scheduled on a single machine.
The goal is to schedule the jobs without preemption such that no job starts before its release time and no job completes much later than its deadline.
The formal definition of the decision variant is as follows.
\begin{definition}[Non-preemptive single machine scheduling with release times and deadlines]
Given a single machine and a set of jobs $J = \{([r_i,d_i],p_i)\mid i=1,\ldots,n\}$,
where $r_i, d_i \in \mathbb{Z}$ are the job's release time and deadline, together forming the job's availability interval $[r_i,d_i]$, and $p_i \in \mathbb{N}$ is the job's processing time.

\textbf{Question:} does there exist a \emph{schedule}, that is, an assignment of start times $t : \{1,\ldots,n\} \to \mathbb{R}$
to the jobs, such that $r_i \leq t(i) \leq d_i - p_i$ for all $i=1,\ldots,n$, and the set of execution intervals~$\{[t(i),t(i)+p_i) \mid i=1,\ldots,n\}$ are pairwise disjoint?
\end{definition}

The $L_{\max}$ optimization criterion asks to minimize the maximum lateness, that is, the maximum difference in time between a job's completion time and its deadline.

The problem is NP-complete by an easy reduction from 3-Partition~\citep{book_scheduling_pinedo},
but branch and bound algorithms work well in practice, at least on certain distributions of randomly generated instances of up to 1000 jobs~\citep{Carlier198242}.

We study a parameterized version of the problem, which has the set of job lengths~$P(S) = \{p_i \mid ([r_i,d_i],p_i) \in S\}$ as a parameter.
The case of unit job lengths ($p_i=1$) is solved by the greedy Earliest Due Date (EDD) algorithm
and the general case of identical job lengths ($p_i=p$) is also solvable in polynomial time~\citep{DBLP:journals/siamcomp/GareyJST81}.
The generalization to~$p_i\in\{1,p\}$ can be solved using a linear programming formulation~\citep{DBLP:conf/esa/Sgall12}.
This approach computes a sequence of starting times such that the length-$p$ can be assigned to these starting times,
with additional constraints that guarantee sufficient idle time for the unit length jobs.
These additional constraints simply provide a lower bound on the amount of idle time left by the schedule of the long jobs over a number of intervals, as if the unit length jobs were preemptive.
This is sufficient because release times, deadlines and processing times are integers so by discretization starting times can always be integers.
If the smaller job length is non-unit, this does not work anymore.
For the case~$\{2q,q\}$, a branch-and-bound algorithm specifically for two job lengths is claimed to run in pseudo-polynomial time~\citep{DBLP:journals/anor/Vakhania04}.

The complexity of the multi-machine version was studied by~\cite{DBLP:journals/siamcomp/SimonsW89}:
they give a polynomial time algorithm for the identical job lengths case and prove the~$\{1,p\}$-case for multiple machines to be NP-complete.
Their proof assumes that both the long job length~$p$ and the number of machines are part of the input.
If we add precedence constraints to the problem, in the identical job lengths case this does not make the problem harder~\citep{DBLP:journals/siamcomp/GareyJST81},
but even the~$\{1,p\}$-problem on one machine with precedences is NP-complete~\citep{trvoor1pmetprecsNPcomplete}, also assuming that the long job length~$p$ is part of the input.

The complexity status of the general two-job-lengths problem has been noted as an open problem \citep{DBLP:journals/siamcomp/SimonsW89,DBLP:conf/esa/Sgall12}.
In this paper we prove NP-completeness of the problem for any fixed pair of non-unit job lengths.\footnote{In Jan Elffers' master thesis, another proof is given that the scheduling problem with two non-unit job lengths is NP-complete, but the reduction given there is more complicated than the one in this paper, because it contains more types of gadgets. The thesis is available at \url{http://resolver.tudelft.nl/uuid:5def2dbb-67d1-4672-a0b1-561d7dc1a74f}.}

Formally, we have the following result, using $p$ to denote the length of the \emph{long} jobs and $q$ for the \emph{short} jobs.

\begin{theorem}[NP-completeness result]\label{thm:mainresult}
Let $p > q > 1$ be two integer job lengths. The scheduling problem with release times and deadlines on the set of job lengths $P(S) = \{p,q\}$ is NP-complete.
\end{theorem}

This implies that the problem is strongly NP-complete, because we reduce all instances of an NP-complete problem in polynomial time and with polynomial size output to an instance with fixed $p$ and $q$.

\section{Overview of the reduction}
Our proof is via an intermediate problem, denoted by $\operatorname{AUX}(p,q)$.
Informally, this is a scheduling problem where jobs can have both an early and a late deadline and where pairs of jobs can be defined of which at least one needs to meet the early deadline.
We prove that this problem is both polynomial-time reducible to the original scheduling problem and NP-complete by a reduction from Boolean Satisfiability.
Together, this implies Theorem~\ref{thm:mainresult}.

Formally, our lemma's are the following:

\begin{lemma}\label{lem:partA}
For any two integer job lengths $p > q \geq 1$, the problem $\operatorname{AUX}(p,q)$
is polynomial-time reducible to the original scheduling problem on job lengths $\{p,q\}$.
\end{lemma}

\begin{lemma}\label{lem:partB}
For any two integer job lengths $p > q > 1$, the problem $\operatorname{AUX}(p,q)$ is NP-complete.
\end{lemma}


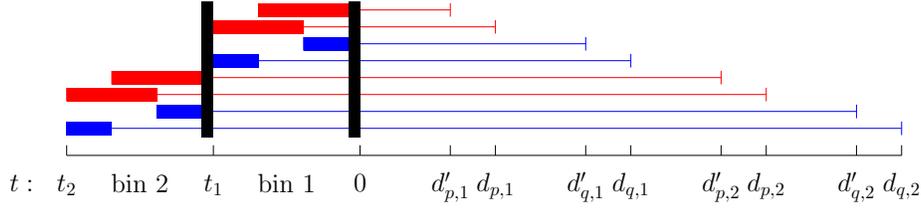
\begin{figure*}
\makebox[\textwidth][c]{
\scalebox{0.125}{
\begin{tikzpicture}[xscale=1.2,yscale=1.8]

\def\dpOneAcc{8}
\def\dpOne{12}
\def\dqOneAcc{20}
\def\dqOne{24}
\def\dpTwoAcc{32}
\def\dpTwo{36}
\def\dqTwoAcc{44}
\def\dqTwo{48}

\draw[color=blue,line width=2pt] (-26.0,8.1) -- (-26.0,8.9);
\draw[color=blue,line width=2pt] (\dqTwo,8.1) -- (\dqTwo,8.9);
\draw[color=blue,line width=2pt] (-26.0,8.5) -- (\dqTwo,8.5);
\fill[color=blue] (-26.0,8.1) -- (-22.0,8.1) -- (-22.0,8.9) -- (-26.0,8.9) -- cycle;
\draw[color=blue,line width=2pt] (-18.0,9.1) -- (-18.0,9.9);
\draw[color=blue,line width=2pt] (\dqTwoAcc,9.1) -- (\dqTwoAcc,9.9);
\draw[color=blue,line width=2pt] (-18.0,9.5) -- (\dqTwoAcc,9.5);
\fill[color=blue] (-18.0,9.1) -- (-14.0,9.1) -- (-14.0,9.9) -- (-18.0,9.9) -- cycle;
\draw[color=red,line width=2pt] (-26.0,10.1) -- (-26.0,10.9);
\draw[color=red,line width=2pt] (\dpTwo,10.1) -- (\dpTwo,10.9);
\draw[color=red,line width=2pt] (-26.0,10.5) -- (\dpTwo,10.5);
\fill[color=red] (-26.0,10.1) -- (-18.0,10.1) -- (-18.0,10.9) -- (-26.0,10.9) -- cycle;
\draw[color=red,line width=2pt] (-22.0,11.1) -- (-22.0,11.9);
\draw[color=red,line width=2pt] (\dpTwoAcc,11.1) -- (\dpTwoAcc,11.9);
\draw[color=red,line width=2pt] (-22.0,11.5) -- (\dpTwoAcc,11.5);
\fill[color=red] (-22.0,11.1) -- (-14.0,11.1) -- (-14.0,11.9) -- (-22.0,11.9) -- cycle;
\draw[color=blue,line width=2pt] (-13.0,12.1) -- (-13.0,12.9);
\draw[color=blue,line width=2pt] (\dqOne,12.1) -- (\dqOne,12.9);
\draw[color=blue,line width=2pt] (-13.0,12.5) -- (\dqOne,12.5);
\fill[color=blue] (-13.0,12.1) -- (-9.0,12.1) -- (-9.0,12.9) -- (-13.0,12.9) -- cycle;
\draw[color=blue,line width=2pt] (-5.0,13.1) -- (-5.0,13.9);
\draw[color=blue,line width=2pt] (\dqOneAcc,13.1) -- (\dqOneAcc,13.9);
\draw[color=blue,line width=2pt] (-5.0,13.5) -- (\dqOneAcc,13.5);
\fill[color=blue] (-5.0,13.1) -- (-1.0,13.1) -- (-1.0,13.9) -- (-5.0,13.9) -- cycle;
\draw[color=red,line width=2pt] (-13.0,14.1) -- (-13.0,14.9);
\draw[color=red,line width=2pt] (\dpOne,14.1) -- (\dpOne,14.9);
\draw[color=red,line width=2pt] (-13.0,14.5) -- (\dpOne,14.5);
\fill[color=red] (-13.0,14.1) -- (-5.0,14.1) -- (-5.0,14.9) -- (-13.0,14.9) -- cycle;
\draw[color=red,line width=2pt] (-9.0,15.1) -- (-9.0,15.9);
\draw[color=red,line width=2pt] (\dpOneAcc,15.1) -- (\dpOneAcc,15.9);
\draw[color=red,line width=2pt] (-9.0,15.5) -- (\dpOneAcc,15.5);
\fill[color=red] (-9.0,15.1) -- (-1.0,15.1) -- (-1.0,15.9) -- (-9.0,15.9) -- cycle;
\fill[color=black] (-1.0,8.0) -- (0.0,8.0) -- (0.0,16.0) -- (-1.0,16.0) -- cycle;
\fill[color=black] (-14.0,8.0) -- (-13.0,8.0) -- (-13.0,16.0) -- (-14.0,16.0) -- cycle;

\path (-13,8-3.25) node[anchor=base] {\fontsize{16pt}{16pt}\selectfont{\scalebox{5.0}{$t_1$}}};
\path (-26,8-3.25) node[anchor=base] {\fontsize{16pt}{16pt}\selectfont{\scalebox{5.0}{$t_2$}}};
\path (  0,8-3.25) node[anchor=base] {\fontsize{16pt}{16pt}\selectfont{\scalebox{5.0}{$0$}}};
\path (-30,8-3.25) node[anchor=base] {\fontsize{16pt}{16pt}\selectfont{\scalebox{5.0}{$t:$}}};

\draw[color=black] (-26.000000,8-1.100000) -- (-26.000000,8-0.500000);
\draw[color=black] (-26.000000,8-1.100000) -- (-13.000000,8-1.100000);
\draw[color=black] (-13.000000,8-1.100000) -- (-13.000000,8-0.500000);
\draw[color=black] (-13.000000,8-1.100000) -- (-0.000000,8-1.100000);
\draw[color=black] (-0.000000,8-1.100000) -- (-0.000000,8-0.500000);
\draw[color=black] (-0.000000,8-1.100000) -- (\dqTwo,8-1.100000);
\draw[color=black] (\dpOneAcc,8-1.100000) -- (\dpOneAcc,8-0.500000);
\draw[color=black] (\dpOne,8-1.100000) -- (\dpOne,8-0.500000);
\draw[color=black] (\dqOneAcc,8-1.100000) -- (\dqOneAcc,8-0.500000);
\draw[color=black] (\dqOne,8-1.100000) -- (\dqOne,8-0.500000);
\draw[color=black] (\dpTwoAcc,8-1.100000) -- (\dpTwoAcc,8-0.500000);
\draw[color=black] (\dpTwo,8-1.100000) -- (\dpTwo,8-0.500000);
\draw[color=black] (\dqTwoAcc,8-1.100000) -- (\dqTwoAcc,8-0.500000);
\draw[color=black] (\dqTwo,8-1.100000) -- (\dqTwo,8-0.500000);
\path (\dpOneAcc,8-3.25) node[anchor=base] {\fontsize{16pt}{16pt}\selectfont{\scalebox{5.0}{$d'_{p,1}$}}};
\path (\dpOne,8-3.25) node[anchor=base] {\fontsize{16pt}{16pt}\selectfont{\scalebox{5.0}{$d_{p,1}$}}};
\path (\dqOneAcc,8-3.25) node[anchor=base] {\fontsize{16pt}{16pt}\selectfont{\scalebox{5.0}{$d'_{q,1}$}}};
\path (\dqOne,8-3.25) node[anchor=base] {\fontsize{16pt}{16pt}\selectfont{\scalebox{5.0}{$d_{q,1}$}}};
\path (\dpTwoAcc,8-3.25) node[anchor=base] {\fontsize{16pt}{16pt}\selectfont{\scalebox{5.0}{$d'_{p,2}$}}};
\path (\dpTwo,8-3.25) node[anchor=base] {\fontsize{16pt}{16pt}\selectfont{\scalebox{5.0}{$d_{p,2}$}}};
\path (\dqTwoAcc,8-3.25) node[anchor=base] {\fontsize{16pt}{16pt}\selectfont{\scalebox{5.0}{$d'_{q,2}$}}};
\path (\dqTwo,8-3.25) node[anchor=base] {\fontsize{16pt}{16pt}\selectfont{\scalebox{5.0}{$d_{q,2}$}}};

\node[draw=none,fill=none,anchor=mid] at (-19.5,8-3.25) {\fontsize{16pt}{16pt}\selectfont{\scalebox{5.0}{bin 2}}};

\node[draw=none,fill=none,anchor=mid] at (-6.5,8-3.25) {\fontsize{16pt}{16pt}\selectfont{\scalebox{5.0}{bin 1}}};

\end{tikzpicture}
}
}
\caption{\label{fig:highlevel}
An instance of the stacked scheduling problem with~$N=4$ pairs of jobs per stack and without ordinary jobs.
The horizontal lines indicate availability intervals, and the horizontal bars indicate the earliest possible position at which the jobs can be scheduled.
The long (length-$p$) jobs are colored red and the short (length-$q$) jobs are colored blue.
The vertical bars denote separator jobs occupying fixed time intervals.
Note that the availability intervals within each stack form a nested sequence, and the long jobs have earlier deadlines than the short jobs for the same stack index.
}
\end{figure*}

The auxiliary problem $\operatorname{AUX}(p,q)$ is the following extension of the original problem.
Next to a set of jobs~$J = \{([r_i,d_i],p_i)\mid i=1,\ldots,|J|\}$ with non-negative release times $r_i$,
we add two sequences of ``pending'' jobs $J_p, J_q$ of equal size $N = |J_p| = |J_q|$, one per job length, both ordered by deadline.
Pending jobs have a common release time of 0, and two deadlines per job, an early deadline and a late deadline.
For both $J_p$ and $J_q$, the intervals formed by early and late deadline of all pending jobs in the same sequence should not intersect.
Furthermore, there exists an ordering such that the $i$'th most urgent (i.e., with the earliest deadline) long job is more urgent than the $i$'th most urgent short job, for all $1\leq i\leq N$.
The problem is to find a feasible schedule for~$J \cup J_p \cup J_q$ in which
at least one of the $i$'th most urgent long job or the $i$'th most urgent short job completes by its early deadline, for all $1\leq i\leq N$.
The key property of this model is that it is possible to specify dependencies between jobs with deadlines far away from each other.

We choose~$t=0$ as the common release time of all pending jobs and assume the ordinary jobs to have non-negative release times.
The formal definition is as follows.

\begin{definition}[Problem $\operatorname{AUX}(p,q)$]\label{defn:model}
Given are a set of \emph{ordinary} jobs~$J$ with long job lengths $p$ or short lengths $q$, and non-negative release times and deadlines,
and two sequences of $N$ \emph{pending} jobs, $J_p$ and $J_q$.
These pending jobs have release time $0$, and are ordered by deadline in the sequence, so for each $1\leq i\leq N$, the pending jobs with index $i$ from $J_p$ and $J_q$ are 
$J_{p,i}=([0,d_{p,i}],p)$ and $J_{q,i}=([0,d_{q,i}],q)$, respectively.
Additionally, each pending job has an early deadline, denoted by $d'_{p,i}$ and $d'_{q,i}$, respectively, such that these deadlines meet the following conditions:
\begin{align*}
d'_{p,1} \leq d_{p,1} \leq d'_{p,2} \leq d_{p,2} \leq \ldots \leq d'_{p,N} \leq d_{p,N} \\
d'_{q,1} \leq d_{q,1} \leq d'_{q,2} \leq d_{q,2} \leq \ldots \leq d'_{q,N} \leq d_{q,N} \\
d_{p,i} \leq d'_{q,i} \quad \text{for all~$i=1,\ldots,N$}
\end{align*}
The \emph{urgency order} on these job pairs is defined as $(J_{p,1}, J_{q,1}), (J_{p,2}, J_{q,2}), \ldots, (J_{p,n}, J_{q,n})$.

\textbf{Question:} does there exist a feasible schedule for~$J \cup J_p \cup J_q$ such that for all~$1\leq i\leq N$, at least one of the jobs $J_{p,i}$ or $J_{q,i}$ completes by the early deadline, i.e., 
$J_{p,i}$ completes by time~$\leq d'_{p,i}$ or~$J_{q,i}$ completes by time~$\leq d'_{q,i}$?
\begin{align*}
\end{align*}
\end{definition}

We say that a pair of a long and a short pending job with the same position $i$ in the sequence (i.e., the same urgency) are \emph{connected}.
In a feasible schedule, only one of each such connected pairs can be scheduled after the early deadline (but before the later deadline). 
We say that that job is then \emph{late}, otherwise we say that the job is \emph{early}.

\section{Reduction from~$\operatorname{AUX}(p,q)$ to the original scheduling problem}
In this section we prove Lemma~\ref{lem:partA}, that is, we prove that $\operatorname{AUX}(p,q)$ is efficiently reducible to the original problem.
We call the problem defined by the reduction the \emph{stacked scheduling problem}.
This is a class of instances of two job lengths.
First we define this reduction, and then we show that each feasible schedule for an instance of $\operatorname{AUX}(p,q)$ can be translated to a feasible schedule for an instance $I$ of this stacked scheduling problem, and vice versa.

First we sketch the idea of the transformation informally.
For each pending job (which has an early and late deadline), we create two jobs with ordinary availability intervals, which we call the \emph{inner} and the \emph{outer} job, together forming a \emph{stack pair}.
The inner job has the early deadline~$d'$ and the outer job has the late deadline~$d$.
The availability intervals of these jobs are extended to before~$t=0$ such that the intervals form a nested sequence per job length.
We add \emph{separator} jobs (jobs occupying fixed positions) before~$t=0$, partitioning the time line before~$t=0$ into \emph{bins} of length~$p+q$.
We number the bins~$1,\ldots,N$ starting from~$t=0$ backward in time.
The four jobs representing the $i$'th most urgent long job and the $i$'th most urgent short job are called a \emph{quad}.
The release times of these jobs are aligned with the bins:
the jobs in the $i$'th quad have release times in the~$i$'th bin;
the outer jobs have release times at the start of the bin,
and the inner jobs have release times such that they can only be scheduled at the end of their bin (or later).

An example of a reduced instance of four pairs of pending jobs is displayed in Figure~\ref{fig:highlevel}.
Here the four respective quads are placed above each other, the one with the smallest availability intervals first.

We denote by~$w$ the job length of the separator jobs before $t=0$. The purpose of this is to improve the readability.
For the resulting instance to have job lengths in~$\{p,q\}$, we can choose $w=p$ or $w=q$.
Formally, the reduction is defined as follows.

\begin{definition}[Reduction]\label{defn:redpartA}
Given an instance of $\operatorname{AUX}(p,q)$, the instance $I$ of a $1|r_i|L_{\max}$ scheduling problem comprises of all the ordinary jobs $J$, and the following sequences of $N$ jobs to replace the pending jobs:
separator jobs $J_{\mathit{sep}}$, representatives $J^I_p$ and $J^O_p$ for $J_p$, and representatives $J^I_q$ and $J^O_q$ for $J_q$.
For~$i=1,\ldots,N$, let~$t_i=-(p+q+w)\cdot i$ and then let:
\begin{align*}
J_{\mathit{sep},i} &= ([t_i+p+q,t_i+p+q+w],w) \\
J^I_{p,i} &= ([t_i+q,d'_{p,i}],p) \\
J^O_{p,i} &= ([t_i,d_{p,i}],p) \\
J^I_{q,i} &= ([t_i+p,d'_{q,i}],q) \\
J^O_{q,i} &= ([t_i,d_{q,i}],q) \\
\end{align*}
\end{definition}

To extend a feasible schedule for $\operatorname{AUX}(p,q)$ to a feasible schedule for $I$,
we replace the pending jobs by inner/outer jobs and schedule the remaining inner/outer jobs before $t=0$.
More precisely, for each $i$'th quad, we replace one $i$'th pending job completing by its early deadline by the inner job of its stack pair,
and the other $i$'th pending job by the corresponding outer job.
Consequently, the remaining jobs from each $i$'th quad can be scheduled in bin~$i$.

For the other direction, we prove that any schedule for $I$
can be transformed into the form described above, so that the part after $t=0$ is a schedule for the $\operatorname{AUX}(p,q)$ instance.

Formally, the intended way to split each quad is defined as follows.
\begin{definition}[Proper pairing]\label{defn:properlabel}
For each quad with index $i$, the \emph{proper pairing} consists of the following two pairs: $(J^I_{p,i},J^O_{q,i})$ and $(J^O_{p,i},J^I_{q,i})$.
\end{definition}

The proof idea is that for any other partial schedule until~$t=0$, the load induced on the part after~$t=0$ is both coarser (a long job allows less flexibility than multiple short jobs with the same deadline)
and the load profile induced by the deadlines leans more towards earlier deadlines.

\begin{proof}[Proof of Lemma~\ref{lem:partA}]
A solution to $\operatorname{AUX}(p,q)$ can be translated to a feasible schedule for the instance of Definition~\ref{defn:redpartA}
by doing the following for each quad~$i$:
if $J_{p,i}$ completes early, schedule $J_{p,i}^O$ and $J_{q,i}^I$ in bin~$i$ ($J_{p,i}^O$ to the left),
$J_{p,i}^I$ at the place of $J_{p,i}$ and $J_{q,i}^O$ at the place of $J_{q,i}$;
if $J_{q,i}$ completes early, schedule $J_{q,i}^O$ and $J_{p,i}^I$ in bin~$i$ ($J_{q,i}^O$ to the left),
$J_{q,i}^I$ at the place of $J_{q,i}$ and $J_{p,i}^O$ at the place of $J_{p,i}$.

For the other direction, let a feasible schedule for $I$ be given.
We transform this schedule such that each bin~$i$ contains a pair of jobs from the proper pairing of quad~$i$.
The part of the schedule after $t=0$, replacing $J_{p,i}^I/J_{p,i}^O$ by $J_{p,i}$ and $J_{q,i}^I/J_{q,i}^O$ by $J_{q,i}$, is then a feasible schedule for $\operatorname{AUX}(p,q)$.

First, because all ordinary jobs have release times after~$t=0$, they are not scheduled in the bins.
We do an exchange argument that fills the bins in the required way.
We prove by induction that for~$i=1,\ldots,N$, bins~$1,\ldots,i$ can be filled in this way.
In step~$i$, note that stack pair with index~$i$ is scheduled in bin~$i$ or after~$t=0$,
because by induction bins~$j=1,\ldots,i-1$ are filled with stack pairs with lower index.
We say that we \emph{swap in} a job into bin~$i$ if we add it to the left in the bin and push the other jobs in the bin to the right.
This pushing to the right is possible because all stack job deadlines are after~$t=0$.

If bin~$i$ contains no long job, swap in the~$i$'th outer long job which is scheduled after~$t=0$.
Because the bin has length~$p+q$, one short job can remain and the other(s), with total length at most~$p$, fit in the original position of the long job.
These jobs complete by their deadlines because~$d'_{q,i} \geq d_{p,i}$ and the other short jobs have later deadlines.
If bin~$i$ contains no short job, we can clearly swap in the~$i$'th outer short job.

Now that bin~$i$ contains a long and a short job,
we can swap the leftmost in the bin to an~$i$'th outer job, and the other job to an~$i$'th inner job,
because the availability intervals per stack are nested and those jobs are the ones with the smallest availability interval in their stacks that can be scheduled at these positions.
Therefore the content of bin $i$ is a pair from the proper pairing.
Concluding, we have thus obtained a feasible schedule for $\operatorname{AUX}(p,q)$, completing the proof in the other direction.
\end{proof}

We conclude that the $1|r_i|L_{\max}$ scheduling problem on job lengths $\{p,q\}$ is as hard as $\operatorname{AUX}(p,q)$.

\section{NP-completeness of $\operatorname{AUX}(p,q)$ for $q>1$}
In this section we prove Lemma~\ref{lem:partB}, that is, we prove the auxiliary scheduling problem to be NP-complete if $p > q > 1$.

We reduce from the Boolean Satisfiability (SAT) problem, assuming an instance is given in conjunctive normal form.
\begin{definition}[Boolean Satisfiability]
Given is a Boolean formula which contains $n$ variables $x_1,x_2,\ldots,x_n$ and $m$ clauses, where each clause $C_j$ is a disjunction of literals (where a literal is a variable $x_i$ or a negation of a variable, $\neg x_i$).

\textbf{Question:} does there exist a satisfying assignment of true or false to each of the variables such that in each clause at least one of the literals evaluates to true?
\end{definition}

We first explain the idea of the reduction by a simple example. 
Then we give the general reduction and proof.

\subsection{Informal description}
Let us first consider an example of a reduction of a Boolean formula with two variables, $x_1$ and $x_2$, and two clauses, $C_1 = (x_1\vee x_2)$ and $C_2 = (\neg x_1 \vee x_2)$.
As an example of the general reduction, we first show how to reduce this instance of the satisfiability problem to $\operatorname{AUX}(4,2)$.
This scheduling problem is explained in the following step by step.
On the most abstract level we define sections, within each section we define a number of blocks, and each block consists of two or three jobs.

\paragraph{Reduction}
The reduction consists of four \emph{sections} of equal length, each corresponding to a literal, in the order $x_1$, $\neg{x}_1$, $x_2$, $\neg{x}_2$.
Each section corresponds to a time interval which ends with a separator job.
In each section, we have a sequence of \emph{blocks} positioned after each other.
A block is a set of pending jobs and ordinary jobs with deadlines close to each other.
In this example, each section for literal~$l$ has three blocks: one for the literal, $\Lit{l}$, and one for each clause, i.e., $\Cl{l}{1}$ and $\Cl{l}{2}$.
In general each section has $m+1$ blocks, where $m$ is the number of clauses. All clauses are represented at each literal section.
Each section has exactly one unit of idle time.

The four types of blocks we use are given in Figure~\ref{fig:hetplaatje}.
The $\Lit{l}$ block has type $V^+$ for positive literals ($x_1$, $x_2$) and type $V^-$ for negative literals ($\neg{x}_1$, $\neg{x}_2$).
The $\Cl{l}{j}$ block has type $\Cactive$ if $l \in C_j$ and type $\Cinactive$ otherwise.
A $\Cactive$ block has a long and a short pending job with deadlines 6 and 7 relative to the earliest start time, and a $\Cinactive$ block also has a long and a short pending job with deadlines 5 and 7 relative to the earliest start time.
These blocks are formally defined in Definition~\ref{def:blocks-def}.
The deadlines of jobs in the $\Cl{l}{j}$ blocks are shifted in time by the sum of job lengths in the blocks preceding it within the same section.

Pending jobs form pairs as follows: each long pending job in a literal block is paired with the short pending job in the next literal block, and each long pending job in a clause block is paired with the short pending job in the next block for the same clause.
This leaves the first $m$ short pending jobs and the last $m$ long pending jobs unpaired, and these are required to finish before their early deadline, effectively turning them into ordinary jobs.

The resulting $\operatorname{AUX}(4,2)$ scheduling instance is visualized in Figure~\ref{fig:helereductie}.
In this figure, each job is positioned at its earliest possible start time among all feasible schedules. For pending jobs, this is later than the release time of 0, because other jobs occupy the earlier time interval (Proposition~\ref{prop:blk}).
Separator jobs are represented here by the black vertical bars.

\begin{figure*}\centering
\begin{minipage}[b]{.5\linewidth}
\centering
\begin{tikzpicture}[xscale=0.25,yscale=0.5]
\fill[color=blue, opacity=0.1] (0,4) -- (18,4) -- (18,1) -- (0,1) -- cycle;

	\node[draw=none,fill=none,yshift=-10pt,anchor=base] at (-1.5,.75) {  $t:$};
	\draw[color=gray] (0.0,1.0) -- (0.0,4.0);
	\node[draw=none,fill=none,yshift=-10pt,anchor=base] at (0,.75) {  $0$};
	\draw[color=gray] (4.0,1.0) -- (4.0,4.0);
	\node[draw=none,fill=none,yshift=-10pt,anchor=base] at (4,.75) {  $2$};
	\draw[color=gray] (8.0,1.0) -- (8.0,4.0);
	\node[draw=none,fill=none,yshift=-10pt,anchor=base] at (8,.75) {  $4$};
	\draw[color=gray] (12.0,1.0) -- (12.0,4.0);
	\node[draw=none,fill=none,yshift=-10pt,anchor=base] at (12,.75) {  $6$};
	\draw[color=gray] (16.0,1.0) -- (16.0,4.0);
	\node[draw=none,fill=none,yshift=-10pt,anchor=base] at (16,.75) {  $8$};
	\draw[color=black] (20.0,1.0) -- (20.0,4.0);
	\node[draw=none,fill=none,yshift=-10pt,anchor=base] at (20,.75) {  $10$};

	\draw[line width=8pt,color=blue] (2.0,3.5) -- (6.0,3.5);
	\draw[line width=2pt,color=blue] (2.0,3.5) -- (8.0,3.5);
	\draw[line width=2pt] (2.0,3.1) -- (2.0,3.9);
	\draw[line width=2pt] (8.0,3.1) -- (8.0,3.9);
	\draw[line width=8pt,color=blue] (0.0,2.5) -- (4,2.5);
	\draw[line width=2pt,color=blue] (0.0,2.5) -- (18,2.5);
	\draw[line width=2pt] (0.0,2.1) -- (0.0,2.9);
	\draw[line width=2pt] (18.0,2.1) -- (18.0,2.9);
	\draw[line width=2pt,color=red] (-2.0,1.5) -- (16.0,1.5);
	\draw[line width=8pt,color=red] (0,1.5) -- (8.0,1.5);
	\draw[line width=2pt] (14.0,1.1) -- (14.0,1.9);
	\draw[line width=2pt] (16.0,1.1) -- (16.0,1.9);
	\end{tikzpicture}
\subcaption{The $V^+$ block.}
\end{minipage}%
\begin{minipage}[b]{.5\linewidth}
	\centering
	\begin{tikzpicture}[xscale=0.25,yscale=0.5]
\fill[color=blue, opacity=0.1] (0,3) -- (14,3) -- (14,1) -- (0,1) -- cycle;

	\node[draw=none,fill=none,yshift=-10pt,anchor=base] at (-1.5,.75) {  $t:$};
	\draw[color=gray] (0.0,1.0) -- (0.0,4.0);
	\node[draw=none,fill=none,yshift=-10pt,anchor=base] at (0,.75) {  $0$};
	\draw[color=gray] (4.0,1.0) -- (4.0,4.0);
	\node[draw=none,fill=none,yshift=-10pt,anchor=base] at (4,.75) {  $2$};
	\draw[color=gray] (8.0,1.0) -- (8.0,4.0);
	\node[draw=none,fill=none,yshift=-10pt,anchor=base] at (8,.75) {  $4$};
	\draw[color=gray] (12.0,1.0) -- (12.0,4.0);
	\node[draw=none,fill=none,yshift=-10pt,anchor=base] at (12,.75) {  $6$};
	\draw[color=black] (16.0,1.0) -- (16.0,4.0);
	\node[draw=none,fill=none,yshift=-10pt,anchor=base] at (16,.75) {  $8$};

	\draw[line width=8pt,color=blue] (0.0,1.5) -- (4,1.5);
	\draw[line width=2pt,color=blue] (-2.0,1.5) -- (14,1.5);
	\draw[line width=2pt] (12.0,1.1) -- (12.0,1.9);
	\draw[line width=2pt] (14.0,1.1) -- (14.0,1.9);
	\draw[line width=8pt,color=red] (0.0,2.5) -- (8,2.5);
	\draw[line width=2pt,color=red] (-2.0,2.5) -- (14,2.5);
	\draw[line width=2pt] (12.0,2.1) -- (12.0,2.9);
	\draw[line width=2pt] (14.0,2.1) -- (14.0,2.9);
	\end{tikzpicture}
\subcaption{The $\Cactive$ block.}
\end{minipage}
%
\\
\vspace{.5cm}
\begin{minipage}[b]{.5\linewidth}
	\centering
	\begin{tikzpicture}[xscale=0.25,yscale=0.5]
\fill[opacity=0.1] (0,4) -- (18,4) -- (18,1) -- (0,1) -- cycle;

	\node[draw=none,fill=none,yshift=-10pt,anchor=base] at (-1.5,.75) {  $t:$};
	\draw[color=gray] (0.0,1.0) -- (0.0,4.0);
	\node[draw=none,fill=none,yshift=-10pt,anchor=base] at (0,.75) {  $0$};
	\draw[color=gray] (4.0,1.0) -- (4.0,4.0);
	\node[draw=none,fill=none,yshift=-10pt,anchor=base] at (4,.75) {  $2$};
	\draw[color=gray] (8.0,1.0) -- (8.0,4.0);
	\node[draw=none,fill=none,yshift=-10pt,anchor=base] at (8,.75) {  $4$};
	\draw[color=gray] (12.0,1.0) -- (12.0,4.0);
	\node[draw=none,fill=none,yshift=-10pt,anchor=base] at (12,.75) {  $6$};
	\draw[color=gray] (16.0,1.0) -- (16.0,4.0);
	\node[draw=none,fill=none,yshift=-10pt,anchor=base] at (16,.75) {  $8$};
	\draw[color=black] (20.0,1.0) -- (20.0,4.0);
	\node[draw=none,fill=none,yshift=-10pt,anchor=base] at (20,.75) {  $10$};

	\draw[line width=8pt,color=blue] (6.0,3.5) -- (10.0,3.5);
	\draw[line width=2pt,color=blue] (6.0,3.5) -- (12.0,3.5);
	\draw[line width=2pt] (6.0,3.1) -- (6.0,3.9);
	\draw[line width=2pt] (12.0,3.1) -- (12.0,3.9);
	\draw[line width=8pt,color=red] (0.0,2.5) -- (8,2.5);
	\draw[line width=2pt,color=red] (0.0,2.5) -- (18,2.5);
	\draw[line width=2pt] (0.0,2.1) -- (0.0,2.9);
	\draw[line width=2pt] (18.0,2.1) -- (18.0,2.9);
	\draw[line width=2pt,color=blue] (-2.0,1.5) -- (16.0,1.5);
	\draw[line width=8pt,color=blue] (0,1.5) -- (4.0,1.5);
	\draw[line width=2pt] (4.0,1.1) -- (4.0,1.9);
	\draw[line width=2pt] (16.0,1.1) -- (16.0,1.9);
	\end{tikzpicture}
\subcaption{The $V^-$ block.}
\end{minipage}%
\begin{minipage}[b]{.5\linewidth}
	\centering
	\begin{tikzpicture}[xscale=0.25,yscale=0.5]
\fill[opacity=0.1] (0,3) -- (14,3) -- (14,1) -- (0,1) -- cycle;

	\node[draw=none,fill=none,yshift=-10pt,anchor=base] at (-1.5,.75) {  $t:$};
	\draw[color=gray] (0.0,1.0) -- (0.0,4.0);
	\node[draw=none,fill=none,yshift=-10pt,anchor=base] at (0,.75) {  $0$};
	\draw[color=gray] (4.0,1.0) -- (4.0,4.0);
	\node[draw=none,fill=none,yshift=-10pt,anchor=base] at (4,.75) {  $2$};
	\draw[color=gray] (8.0,1.0) -- (8.0,4.0);
	\node[draw=none,fill=none,yshift=-10pt,anchor=base] at (8,.75) {  $4$};
	\draw[color=gray] (12.0,1.0) -- (12.0,4.0);
	\node[draw=none,fill=none,yshift=-10pt,anchor=base] at (12,.75) {  $6$};
	\draw[color=black] (16.0,1.0) -- (16.0,4.0);
	\node[draw=none,fill=none,yshift=-10pt,anchor=base] at (16,.75) {  $8$};

	\draw[line width=8pt,color=blue] (0.0,1.5) -- (4,1.5);
	\draw[line width=2pt,color=blue] (-2.0,1.5) -- (14,1.5);
	\draw[line width=2pt] (10.0,1.1) -- (10.0,1.9);
	\draw[line width=2pt] (14.0,1.1) -- (14.0,1.9);
	\draw[line width=8pt,color=red] (0.0,2.5) -- (8,2.5);
	\draw[line width=2pt,color=red] (-2.0,2.5) -- (14,2.5);
	\draw[line width=2pt] (10.0,2.1) -- (10.0,2.9);
	\draw[line width=2pt] (14.0,2.1) -- (14.0,2.9);
	\end{tikzpicture}
\subcaption{The $\Cinactive$ block.}
%
\end{minipage}

\caption{\label{fig:hetplaatje}This visualizes the four types of blocks we use. The rectangles represent jobs and the bars represent their availability intervals: long jobs are colored red, and short jobs are colored blue.
The jobs with two deadlines denote pending jobs. The $V^+$ and $\Cactive$ blocks are displayed on a light blue background, and the $V^-$ and $\Cinactive$ blocks have a light grey background to be able to identify them more easily in subsequent figures.}
\end{figure*}
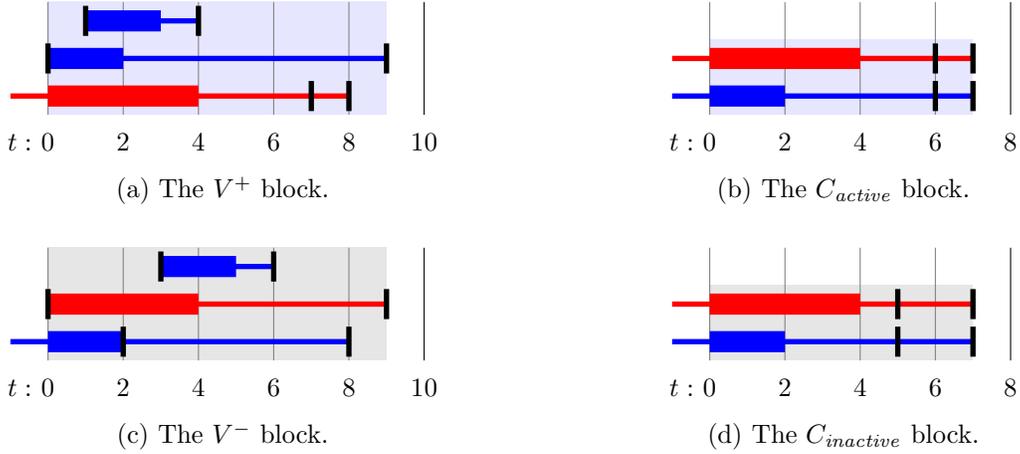

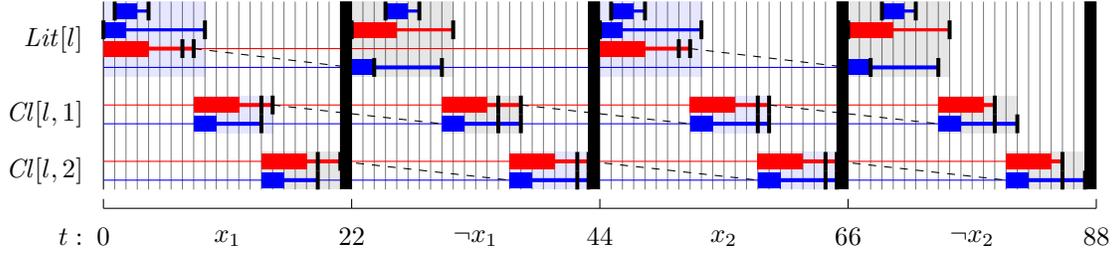
\begin{figure*}\centering
\makebox[\textwidth][c]{
\begin{tikzpicture}[xscale=0.15,yscale=0.25]
\draw[color=black] (8,18.0) -- (8,28);
\draw[color=gray] (9,18.0) -- (9,28);
\draw[color=gray] (10,18.0) -- (10,28);
\draw[color=gray] (11,18.0) -- (11,28);
\draw[color=gray] (12,18.0) -- (12,28);
\draw[color=gray] (13,18.0) -- (13,28);
\draw[color=gray] (14,18.0) -- (14,28);
\draw[color=gray] (15,18.0) -- (15,28);
\draw[color=gray] (16,18.0) -- (16,28);
\draw[color=gray] (17,18.0) -- (17,28);
\draw[color=gray] (18,18.0) -- (18,28);
\draw[color=gray] (19,18.0) -- (19,28);
\draw[color=gray] (20,18.0) -- (20,28);
\draw[color=gray] (21,18.0) -- (21,28);
\draw[color=gray] (22,18.0) -- (22,28);
\draw[color=gray] (23,18.0) -- (23,28);
\draw[color=gray] (24,18.0) -- (24,28);
\draw[color=gray] (25,18.0) -- (25,28);
\draw[color=gray] (26,18.0) -- (26,28);
\draw[color=gray] (27,18.0) -- (27,28);
\draw[color=gray] (28,18.0) -- (28,28);
\draw[color=gray] (29,18.0) -- (29,28);
\draw[color=gray] (30,18.0) -- (30,28);
\draw[color=gray] (31,18.0) -- (31,28);
\draw[color=gray] (32,18.0) -- (32,28);
\draw[color=gray] (33,18.0) -- (33,28);
\draw[color=gray] (34,18.0) -- (34,28);
\draw[color=gray] (35,18.0) -- (35,28);
\draw[color=gray] (36,18.0) -- (36,28);
\draw[color=gray] (37,18.0) -- (37,28);
\draw[color=gray] (38,18.0) -- (38,28);
\draw[color=gray] (39,18.0) -- (39,28);
\draw[color=gray] (40,18.0) -- (40,28);
\draw[color=gray] (41,18.0) -- (41,28);
\draw[color=gray] (42,18.0) -- (42,28);
\draw[color=gray] (43,18.0) -- (43,28);
\draw[color=gray] (44,18.0) -- (44,28);
\draw[color=gray] (45,18.0) -- (45,28);
\draw[color=gray] (46,18.0) -- (46,28);
\draw[color=gray] (47,18.0) -- (47,28);
\draw[color=gray] (48,18.0) -- (48,28);
\draw[color=gray] (49,18.0) -- (49,28);
\draw[color=gray] (50,18.0) -- (50,28);
\draw[color=gray] (51,18.0) -- (51,28);
\draw[color=gray] (52,18.0) -- (52,28);
\draw[color=gray] (53,18.0) -- (53,28);
\draw[color=gray] (54,18.0) -- (54,28);
\draw[color=gray] (55,18.0) -- (55,28);
\draw[color=gray] (56,18.0) -- (56,28);
\draw[color=gray] (57,18.0) -- (57,28);
\draw[color=gray] (58,18.0) -- (58,28);
\draw[color=gray] (59,18.0) -- (59,28);
\draw[color=gray] (60,18.0) -- (60,28);
\draw[color=gray] (61,18.0) -- (61,28);
\draw[color=gray] (62,18.0) -- (62,28);
\draw[color=gray] (63,18.0) -- (63,28);
\draw[color=gray] (64,18.0) -- (64,28);
\draw[color=gray] (65,18.0) -- (65,28);
\draw[color=gray] (66,18.0) -- (66,28);
\draw[color=gray] (67,18.0) -- (67,28);
\draw[color=gray] (68,18.0) -- (68,28);
\draw[color=gray] (69,18.0) -- (69,28);
\draw[color=gray] (70,18.0) -- (70,28);
\draw[color=gray] (71,18.0) -- (71,28);
\draw[color=gray] (72,18.0) -- (72,28);
\draw[color=gray] (73,18.0) -- (73,28);
\draw[color=gray] (74,18.0) -- (74,28);
\draw[color=gray] (75,18.0) -- (75,28);
\draw[color=gray] (76,18.0) -- (76,28);
\draw[color=gray] (77,18.0) -- (77,28);
\draw[color=gray] (78,18.0) -- (78,28);
\draw[color=gray] (79,18.0) -- (79,28);
\draw[color=gray] (80,18.0) -- (80,28);
\draw[color=gray] (81,18.0) -- (81,28);
\draw[color=gray] (82,18.0) -- (82,28);
\draw[color=gray] (83,18.0) -- (83,28);
\draw[color=gray] (84,18.0) -- (84,28);
\draw[color=gray] (85,18.0) -- (85,28);
\draw[color=gray] (86,18.0) -- (86,28);
\draw[color=gray] (87,18.0) -- (87,28);
\draw[color=gray] (88,18.0) -- (88,28);
\draw[color=gray] (89,18.0) -- (89,28);
\draw[color=gray] (90,18.0) -- (90,28);
\draw[color=gray] (91,18.0) -- (91,28);
\draw[color=gray] (92,18.0) -- (92,28);
\draw[color=gray] (93,18.0) -- (93,28);
\draw[color=gray] (94,18.0) -- (94,28);
\draw[color=gray] (95,18.0) -- (95,28);
\draw[color=gray] (96,18.0) -- (96,28);

\draw[color=black] (8.000000,17) -- (30.000000,17);
\draw[color=black] (8.000000,17) -- (8.000000,17.5);
\node[draw=none,fill=none,anchor=mid] at (19,15.5) {$x_1$};
\draw[color=black] (30.000000,17) -- (52.000000,17);
\draw[color=black] (30.000000,17) -- (30.000000,17.5);
\node[draw=none,fill=none,anchor=mid] at (41,15.5) {$\neg x_1$};
\draw[color=black] (52.000000,17) -- (74.000000,17);
\draw[color=black] (52.000000,17) -- (52.000000,17.5);
\node[draw=none,fill=none,anchor=mid] at (63,15.5) {$x_2$};
\draw[color=black] (74.000000,17) -- (96.000000,17);
\draw[color=black] (74.000000,17) -- (74.000000,17.5);
\node[draw=none,fill=none,anchor=mid] at (85,15.5) {$\neg x_2$};
\draw[color=black] (96.000000,17) -- (96.000000,17.5);
\node[draw=none,fill=none,anchor=east] at (7,15.5) {  $t:$};
\node[draw=none,fill=none,anchor=east] at (7,19) {  $\Cl{l}{2}$};
\node[draw=none,fill=none,anchor=east] at (7,22) {  $\Cl{l}{1}$};
\node[draw=none,fill=none,anchor=east] at (7,26) {  $\Lit{l}$};
\node[draw=none,fill=none] at (8,15.5) {  $0$};
\node[draw=none,fill=none] at (30,15.5) {  $22$};
\node[draw=none,fill=none] at (52,15.5) {  $44$};
\node[draw=none,fill=none] at (74,15.5) {  $66$};
\node[draw=none,fill=none] at (96,15.5) {  $88$};

\fill[color=blue, opacity=0.1] (8,28) -- (17,28) -- (17,24) -- (8,24) -- cycle;
\fill[opacity=0.1]             (30,28) -- (39,28) -- (39,24) -- (30,24) -- cycle;
\fill[color=blue, opacity=0.1] (52,28) -- (61,28) -- (61,24) -- (52,24) -- cycle;
\fill[opacity=0.1]             (74,28) -- (83,28) -- (83,24) -- (74,24) -- cycle;

\fill[color=blue, opacity=0.1] (16,23) -- (23,23) -- (23,21) -- (16,21) -- cycle;
\fill[opacity=0.1]             (38,23) -- (45,23) -- (45,21) -- (38,21) -- cycle;
\fill[color=blue, opacity=0.1] (60,23) -- (67,23) -- (67,21) -- (60,21) -- cycle;
\fill[opacity=0.1]             (82,23) -- (89,23) -- (89,21) -- (82,21) -- cycle;

\fill[opacity=0.1]             (22,20) -- (29,20) -- (29,18) -- (22,18) -- cycle;
\fill[color=blue, opacity=0.1] (44,20) -- (51,20) -- (51,18) -- (44,18) -- cycle;
\fill[color=blue, opacity=0.1] (66,20) -- (73,20) -- (73,18) -- (66,18) -- cycle;
\fill[opacity=0.1]             (88,20) -- (95,20) -- (95,18) -- (88,18) -- cycle;

\draw[color=blue,very thick] (9,27.5) -- (12,27.5);
\fill[color=blue] (9,27.1) -- (11,27.1) -- (11,27.9) -- (9,27.9) -- cycle;
\draw[color=black,line width=1.5pt] (9,27) -- (9,28);
\draw[color=black,line width=1.5pt] (12,27) -- (12,28);
\draw[color=blue,very thick] (8,26.5) -- (17,26.5);
\fill[color=blue] (8,26.1) -- (10,26.1) -- (10,26.9) -- (8,26.9) -- cycle;
\draw[color=black,line width=1.5pt] (8,26) -- (8,27);
\draw[color=black,line width=1.5pt] (17,26) -- (17,27);
\draw[color=blue,very thick] (33,27.5) -- (36,27.5);
\fill[color=blue] (33,27.1) -- (35,27.1) -- (35,27.9) -- (33,27.9) -- cycle;
\draw[color=black,line width=1.5pt] (33,27) -- (33,28);
\draw[color=black,line width=1.5pt] (36,27) -- (36,28);
\draw[color=red,very thick] (30,26.5) -- (39,26.5);
\fill[color=red] (30,26.1) -- (34,26.1) -- (34,26.9) -- (30,26.9) -- cycle;
\draw[color=black,line width=1.5pt] (30,26) -- (30,27);
\draw[color=black,line width=1.5pt] (39,26) -- (39,27);
\draw[color=blue,very thick] (53,27.5) -- (56,27.5);
\fill[color=blue] (53,27.1) -- (55,27.1) -- (55,27.9) -- (53,27.9) -- cycle;
\draw[color=black,line width=1.5pt] (53,27) -- (53,28);
\draw[color=black,line width=1.5pt] (56,27) -- (56,28);
\draw[color=blue,very thick] (52,26.5) -- (61,26.5);
\fill[color=blue] (52,26.1) -- (54,26.1) -- (54,26.9) -- (52,26.9) -- cycle;
\draw[color=black,line width=1.5pt] (52,26) -- (52,27);
\draw[color=black,line width=1.5pt] (61,26) -- (61,27);
\draw[color=blue,very thick] (77,27.5) -- (80,27.5);
\fill[color=blue] (77,27.1) -- (79,27.1) -- (79,27.9) -- (77,27.9) -- cycle;
\draw[color=black,line width=1.5pt] (77,27) -- (77,28);
\draw[color=black,line width=1.5pt] (80,27) -- (80,28);
\draw[color=red,very thick] (74,26.5) -- (83,26.5);
\fill[color=red] (74,26.1) -- (78,26.1) -- (78,26.9) -- (74,26.9) -- cycle;
\draw[color=black,line width=1.5pt] (74,26) -- (74,27);
\draw[color=black,line width=1.5pt] (83,26) -- (83,27);

\draw[color=red] (8,25.5) -- (52,25.5);
\draw[color=blue] (8,24.5) -- (74,24.5);
\draw[color=red,line width=1.5pt] (8,25.5) -- (16,25.5);
\fill[color=red] (8,25.1) -- (12,25.1) -- (12,25.9) -- (8,25.9) -- cycle;
\draw[color=black,line width=1.5pt] (15,25.000000) -- (15,26.000000);
\draw[color=black,line width=1.5pt] (16,25.000000) -- (16,26.000000);
\draw[color=blue,line width=1.5pt] (30,24.5) -- (38,24.5);
\fill[color=blue] (30,24.1) -- (32,24.1) -- (32,24.9) -- (30,24.9) -- cycle;
\draw[color=black,line width=1.5pt] (32,24.000000) -- (32,25.000000);
\draw[color=black,line width=1.5pt] (38,24.000000) -- (38,25.000000);
\draw[color=black,dashed] (16,25.5)--(30,24.5);
\draw[color=red,line width=1.5pt] (52,25.5) -- (60,25.5);
\fill[color=red] (52,25.1) -- (56,25.1) -- (56,25.9) -- (52,25.9) -- cycle;
\draw[color=black,line width=1.5pt] (59,25.000000) -- (59,26.000000);
\draw[color=black,line width=1.5pt] (60,25.000000) -- (60,26.000000);
\draw[color=blue,line width=1.5pt] (74,24.5) -- (82,24.5);
\fill[color=blue] (74,24.1) -- (76,24.1) -- (76,24.9) -- (74,24.9) -- cycle;
\draw[color=black,line width=1.5pt] (76,24.000000) -- (76,25.000000);
\draw[color=black,line width=1.5pt] (82,24.000000) -- (82,25.000000);
\draw[color=black,dashed] (60,25.5)--(74,24.5);

\draw[color=red] (8,22.5) -- (82,22.5);
\draw[color=blue] (8,21.5) -- (82,21.5);
\draw[color=red,line width=1.5pt] (16,22.5) -- (23,22.5);
\fill[color=red] (16,22.1) -- (20,22.1) -- (20,22.9) -- (16,22.9) -- cycle;
\draw[color=black,line width=1.5pt] (22,22.000000) -- (22,23.000000);
\draw[color=black,line width=1.5pt] (23,22.000000) -- (23,23.000000);
\draw[color=blue,line width=1.5pt] (16,21.5) -- (22,21.5);
\fill[color=blue] (16,21.1) -- (18,21.1) -- (18,21.9) -- (16,21.9) -- cycle;
\draw[color=black,line width=1.5pt] (22,21.000000) -- (22,22.000000);
\draw[color=black,dashed] (23,22.5)--(38,21.5);

\draw[color=blue,line width=1.5pt] (38,21.5) -- (45,21.5);
\fill[color=blue] (38,21.1) -- (40,21.1) -- (40,21.9) -- (38,21.9) -- cycle;
\draw[color=black,line width=1.5pt] (43,21.000000) -- (43,22.000000);
\draw[color=black,line width=1.5pt] (45,21.000000) -- (45,22.000000);
\draw[color=red,line width=1.5pt] (38,22.5) -- (45,22.5);
\fill[color=red] (38,22.1) -- (42,22.1) -- (42,22.9) -- (38,22.9) -- cycle;
\draw[color=black,line width=1.5pt] (43,22.000000) -- (43,23.000000);
\draw[color=black,line width=1.5pt] (45,22.000000) -- (45,23.000000);
\draw[color=black,dashed] (45,22.5)--(60,21.5);

\draw[color=red,line width=1.5pt] (60,22.5) -- (67,22.5);
\fill[color=red] (60,22.1) -- (64,22.1) -- (64,22.9) -- (60,22.9) -- cycle;
\draw[color=black,line width=1.5pt] (66,22.000000) -- (66,23.000000);
\draw[color=black,line width=1.5pt] (67,22.000000) -- (67,23.000000);
\draw[color=blue,line width=1.5pt] (60,21.5) -- (67,21.5);
\fill[color=blue] (60,21.1) -- (62,21.1) -- (62,21.9) -- (60,21.9) -- cycle;
\draw[color=black,line width=1.5pt] (66,21.000000) -- (66,22.000000);
\draw[color=black,line width=1.5pt] (67,21.000000) -- (67,22.000000);
\draw[color=black,dashed] (67,22.5)--(82,21.5);

\draw[color=blue,line width=1.5pt] (82,21.5) -- (89,21.5);
\fill[color=blue] (82,21.1) -- (84,21.1) -- (84,21.9) -- (82,21.9) -- cycle;
\draw[color=black,line width=1.5pt] (87,21.000000) -- (87,22.000000);
\draw[color=black,line width=1.5pt] (89,21.000000) -- (89,22.000000);
\draw[color=red,line width=1.5pt] (82,22.5) -- (87,22.5);
\fill[color=red] (82,22.1) -- (86,22.1) -- (86,22.9) -- (82,22.9) -- cycle;
\draw[color=black,line width=1.5pt] (87,22.000000) -- (87,23.000000);

\draw[color=red] (8,19.5) -- (88,19.5);
\draw[color=blue] (8,18.5) -- (88,18.5);
\draw[color=blue,line width=1.5pt] (22,18.5) -- (27,18.5);
\fill[color=blue] (22,18.1) -- (24,18.1) -- (24,18.9) -- (22,18.9) -- cycle;
\draw[color=black,line width=1.5pt] (27,18.000000) -- (27,19.000000);
\draw[color=red,line width=1.5pt] (22,19.5) -- (29,19.5);
\fill[color=red] (22,19.1) -- (26,19.1) -- (26,19.9) -- (22,19.9) -- cycle;
\draw[color=black,line width=1.5pt] (27,19.000000) -- (27,20.000000);
\draw[color=black,line width=1.5pt] (29,19.000000) -- (29,20.000000);
\draw[color=black,dashed] (29,19.5)--(44,18.5);

\draw[color=blue,line width=1.5pt] (44,18.5) -- (51,18.5);
\fill[color=blue] (44,18.1) -- (46,18.1) -- (46,18.9) -- (44,18.9) -- cycle;
\draw[color=black,line width=1.5pt] (50,18.000000) -- (50,19.000000);
\draw[color=black,line width=1.5pt] (51,18.000000) -- (51,19.000000);
\draw[color=red,line width=1.5pt] (44,19.5) -- (51,19.5);
\fill[color=red] (44,19.1) -- (48,19.1) -- (48,19.9) -- (44,19.9) -- cycle;
\draw[color=black,line width=1.5pt] (50,19.000000) -- (50,20.000000);
\draw[color=black,line width=1.5pt] (51,19.000000) -- (51,20.000000);
\draw[color=black,dashed] (51,19.5)--(66,18.5);

\draw[color=blue,line width=1.5pt] (66,18.5) -- (73,18.5);
\fill[color=blue] (66,18.1) -- (68,18.1) -- (68,18.9) -- (66,18.9) -- cycle;
\draw[color=black,line width=1.5pt] (72,18.000000) -- (72,19.000000);
\draw[color=black,line width=1.5pt] (73,18.000000) -- (73,19.000000);
\draw[color=red,line width=1.5pt] (66,19.5) -- (73,19.5);
\fill[color=red] (66,19.1) -- (70,19.1) -- (70,19.9) -- (66,19.9) -- cycle;
\draw[color=black,line width=1.5pt] (72,19.000000) -- (72,20.000000);
\draw[color=black,line width=1.5pt] (73,19.000000) -- (73,20.000000);
\draw[color=black,dashed] (73,19.5)--(88,18.5);

\draw[color=blue,line width=1.5pt] (88,18.5) -- (95,18.5);
\fill[color=blue] (88,18.1) -- (90,18.1) -- (90,18.9) -- (88,18.9) -- cycle;
\draw[color=black,line width=1.5pt] (93,18.000000) -- (93,19.000000);
\draw[color=black,line width=1.5pt] (95,18.000000) -- (95,19.000000);
\draw[color=red,line width=1.5pt] (88,19.5) -- (93,19.5);
\fill[color=red] (88,19.1) -- (92,19.1) -- (92,19.9) -- (88,19.9) -- cycle;
\draw[color=black,line width=1.5pt] (93,19.000000) -- (93,20.000000);

\fill[color=black] (29,18) -- (29,28) -- (30,28) -- (30,18) -- cycle;
\fill[color=black] (51,18) -- (51,28) -- (52,28) -- (52,18) -- cycle;
\fill[color=black] (73,18) -- (73,28) -- (74,28) -- (74,18) -- cycle;
\fill[color=black] (95,18) -- (95,28) -- (96,28) -- (96,18) -- cycle;

\end{tikzpicture}
}\caption{\label{fig:helereductie}
This visualizes the auxiliary scheduling problem that is the reduced instance of the example satisfiability problem with two variables and two clauses.
Multiple jobs with release time 0 are positioned in the same row. The thin red or blue line represents the initial availability interval of the last job in the row; for all other jobs, the initial availability is from their release time at 0 until to their (last) deadline.
The thick red or blue lines represent the part of the availability interval that remains for each job in any feasible schedule, and all jobs are shown at their earliest possible start time.
The pairs of long and short pending jobs that form connected pairs are connected by a dashed line.
The first four lines show from left to right the literal blocks~$\Lit{x_1}$, $\Lit{\neg x_1}$, $\Lit{x_2}$, and $\Lit{\neg x_2}$, representing the literals themselves.
The next two lines show the clause blocks~$\Cl{l}{1}$ representing $C_1$; the last two lines show the clause blocks~$\Cl{l}{2}$ representing $C_2$.
In $\Cl{x_1}{1}$ and $\Cl{x_1}{2}$ the short pending job is not paired, and in $\Cl{\neg x_2}{1}$ and $\Cl{\neg x_2}{2}$ the long pending job is not paired.
}
\end{figure*}

\begin{figure*}\centering
\makebox[\textwidth][c]{
\begin{tikzpicture}[xscale=0.15,yscale=0.25]
\draw[color=black] (8,18.0) -- (8,28);
\draw[color=gray] (9,18.0) -- (9,28);
\draw[color=gray] (10,18.0) -- (10,28);
\draw[color=gray] (11,18.0) -- (11,28);
\draw[color=gray] (12,18.0) -- (12,28);
\draw[color=gray] (13,18.0) -- (13,28);
\draw[color=gray] (14,18.0) -- (14,28);
\draw[color=gray] (15,18.0) -- (15,28);
\draw[color=gray] (16,18.0) -- (16,28);
\draw[color=gray] (17,18.0) -- (17,28);
\draw[color=gray] (18,18.0) -- (18,28);
\draw[color=gray] (19,18.0) -- (19,28);
\draw[color=gray] (20,18.0) -- (20,28);
\draw[color=gray] (21,18.0) -- (21,28);
\draw[color=gray] (22,18.0) -- (22,28);
\draw[color=gray] (23,18.0) -- (23,28);
\draw[color=gray] (24,18.0) -- (24,28);
\draw[color=gray] (25,18.0) -- (25,28);
\draw[color=gray] (26,18.0) -- (26,28);
\draw[color=gray] (27,18.0) -- (27,28);
\draw[color=gray] (28,18.0) -- (28,28);
\draw[color=gray] (29,18.0) -- (29,28);
\draw[color=gray] (30,18.0) -- (30,28);
\draw[color=gray] (31,18.0) -- (31,28);
\draw[color=gray] (32,18.0) -- (32,28);
\draw[color=gray] (33,18.0) -- (33,28);
\draw[color=gray] (34,18.0) -- (34,28);
\draw[color=gray] (35,18.0) -- (35,28);
\draw[color=gray] (36,18.0) -- (36,28);
\draw[color=gray] (37,18.0) -- (37,28);
\draw[color=gray] (38,18.0) -- (38,28);
\draw[color=gray] (39,18.0) -- (39,28);
\draw[color=gray] (40,18.0) -- (40,28);
\draw[color=gray] (41,18.0) -- (41,28);
\draw[color=gray] (42,18.0) -- (42,28);
\draw[color=gray] (43,18.0) -- (43,28);
\draw[color=gray] (44,18.0) -- (44,28);
\draw[color=gray] (45,18.0) -- (45,28);
\draw[color=gray] (46,18.0) -- (46,28);
\draw[color=gray] (47,18.0) -- (47,28);
\draw[color=gray] (48,18.0) -- (48,28);
\draw[color=gray] (49,18.0) -- (49,28);
\draw[color=gray] (50,18.0) -- (50,28);
\draw[color=gray] (51,18.0) -- (51,28);
\draw[color=gray] (52,18.0) -- (52,28);
\draw[color=gray] (53,18.0) -- (53,28);
\draw[color=gray] (54,18.0) -- (54,28);
\draw[color=gray] (55,18.0) -- (55,28);
\draw[color=gray] (56,18.0) -- (56,28);
\draw[color=gray] (57,18.0) -- (57,28);
\draw[color=gray] (58,18.0) -- (58,28);
\draw[color=gray] (59,18.0) -- (59,28);
\draw[color=gray] (60,18.0) -- (60,28);
\draw[color=gray] (61,18.0) -- (61,28);
\draw[color=gray] (62,18.0) -- (62,28);
\draw[color=gray] (63,18.0) -- (63,28);
\draw[color=gray] (64,18.0) -- (64,28);
\draw[color=gray] (65,18.0) -- (65,28);
\draw[color=gray] (66,18.0) -- (66,28);
\draw[color=gray] (67,18.0) -- (67,28);
\draw[color=gray] (68,18.0) -- (68,28);
\draw[color=gray] (69,18.0) -- (69,28);
\draw[color=gray] (70,18.0) -- (70,28);
\draw[color=gray] (71,18.0) -- (71,28);
\draw[color=gray] (72,18.0) -- (72,28);
\draw[color=gray] (73,18.0) -- (73,28);
\draw[color=gray] (74,18.0) -- (74,28);
\draw[color=gray] (75,18.0) -- (75,28);
\draw[color=gray] (76,18.0) -- (76,28);
\draw[color=gray] (77,18.0) -- (77,28);
\draw[color=gray] (78,18.0) -- (78,28);
\draw[color=gray] (79,18.0) -- (79,28);
\draw[color=gray] (80,18.0) -- (80,28);
\draw[color=gray] (81,18.0) -- (81,28);
\draw[color=gray] (82,18.0) -- (82,28);
\draw[color=gray] (83,18.0) -- (83,28);
\draw[color=gray] (84,18.0) -- (84,28);
\draw[color=gray] (85,18.0) -- (85,28);
\draw[color=gray] (86,18.0) -- (86,28);
\draw[color=gray] (87,18.0) -- (87,28);
\draw[color=gray] (88,18.0) -- (88,28);
\draw[color=gray] (89,18.0) -- (89,28);
\draw[color=gray] (90,18.0) -- (90,28);
\draw[color=gray] (91,18.0) -- (91,28);
\draw[color=gray] (92,18.0) -- (92,28);
\draw[color=gray] (93,18.0) -- (93,28);
\draw[color=gray] (94,18.0) -- (94,28);
\draw[color=gray] (95,18.0) -- (95,28);
\draw[color=gray] (96,18.0) -- (96,28);

\draw[color=black] (8.000000,17) -- (30.000000,17);
\draw[color=black] (8.000000,17) -- (8.000000,17.5);
\node[draw=none,fill=none,anchor=mid] at (19,15.5) {$x_1$};
\draw[color=black] (30.000000,17) -- (52.000000,17);
\draw[color=black] (30.000000,17) -- (30.000000,17.5);
\node[draw=none,fill=none,anchor=mid] at (41,15.5) {$\neg x_1$};
\draw[color=black] (52.000000,17) -- (74.000000,17);
\draw[color=black] (52.000000,17) -- (52.000000,17.5);
\node[draw=none,fill=none,anchor=mid] at (63,15.5) {$x_2$};
\draw[color=black] (74.000000,17) -- (96.000000,17);
\draw[color=black] (74.000000,17) -- (74.000000,17.5);
\node[draw=none,fill=none,anchor=mid] at (85,15.5) {$\neg x_2$};
\draw[color=black] (96.000000,17) -- (96.000000,17.5);
\node[draw=none,fill=none,anchor=east] at (7,15.5) {  $t:$};
\node[draw=none,fill=none,anchor=east] at (7,19) {  $\Cl{l}{2}$};
\node[draw=none,fill=none,anchor=east] at (7,22) {  $\Cl{l}{1}$};
\node[draw=none,fill=none,anchor=east] at (7,26) {  $\Lit{l}$};
\node[draw=none,fill=none] at (8,15.5) {  $0$};
\node[draw=none,fill=none] at (30,15.5) {  $22$};
\node[draw=none,fill=none] at (52,15.5) {  $44$};
\node[draw=none,fill=none] at (74,15.5) {  $66$};
\node[draw=none,fill=none] at (96,15.5) {  $88$};

\fill[opacity=0.5] (9,28) -- (17,28) -- (17,18) -- (9,18) -- cycle;
\fill[opacity=0.5] (30,28) -- (38,28) -- (38,18) -- (30,18) -- cycle;
\fill[opacity=0.5] (53,28) -- (61,28) -- (61,18) -- (53,18) -- cycle;
\fill[opacity=0.5] (74,28) -- (82,28) -- (82,18) -- (74,18) -- cycle;

\draw[color=blue,very thick] (9,27.5) -- (12,27.5);
\fill[color=blue] (9,27.1) -- (11,27.1) -- (11,27.9) -- (9,27.9) -- cycle;
\draw[color=black,line width=1.5pt] (9,27) -- (9,28);
\draw[color=black,line width=1.5pt] (12,27) -- (12,28);
\draw[color=blue,very thick] (8,26.5) -- (17,26.5);
\fill[color=blue] (15,26.1) -- (17,26.1) -- (17,26.9) -- (15,26.9) -- cycle;
\draw[color=black,line width=1.5pt] (8,26) -- (8,27);
\draw[color=black,line width=1.5pt] (17,26) -- (17,27);
\draw[color=blue,very thick] (33,27.5) -- (36,27.5);
\fill[color=blue] (34,27.1) -- (36,27.1) -- (36,27.9) -- (34,27.9) -- cycle;
\draw[color=black,line width=1.5pt] (33,27) -- (33,28);
\draw[color=black,line width=1.5pt] (36,27) -- (36,28);
\draw[color=red,very thick] (30,26.5) -- (39,26.5);
\fill[color=red] (30,26.1) -- (34,26.1) -- (34,26.9) -- (30,26.9) -- cycle;
\draw[color=black,line width=1.5pt] (30,26) -- (30,27);
\draw[color=black,line width=1.5pt] (39,26) -- (39,27);
\draw[color=blue,very thick] (53,27.5) -- (56,27.5);
\fill[color=blue] (53,27.1) -- (55,27.1) -- (55,27.9) -- (53,27.9) -- cycle;
\draw[color=black,line width=1.5pt] (53,27) -- (53,28);
\draw[color=black,line width=1.5pt] (56,27) -- (56,28);
\draw[color=blue,very thick] (52,26.5) -- (61,26.5);
\fill[color=blue] (59,26.1) -- (61,26.1) -- (61,26.9) -- (59,26.9) -- cycle;
\draw[color=black,line width=1.5pt] (52,26) -- (52,27);
\draw[color=black,line width=1.5pt] (61,26) -- (61,27);
\draw[color=blue,very thick] (77,27.5) -- (80,27.5);
\fill[color=blue] (78,27.1) -- (80,27.1) -- (80,27.9) -- (78,27.9) -- cycle;
\draw[color=black,line width=1.5pt] (77,27) -- (77,28);
\draw[color=black,line width=1.5pt] (80,27) -- (80,28);
\draw[color=red,very thick] (74,26.5) -- (83,26.5);
\fill[color=red] (74,26.1) -- (78,26.1) -- (78,26.9) -- (74,26.9) -- cycle;
\draw[color=black,line width=1.5pt] (74,26) -- (74,27);
\draw[color=black,line width=1.5pt] (83,26) -- (83,27);

\draw[color=red] (8,25.5) -- (52,25.5);
\draw[color=blue] (8,24.5) -- (74,24.5);
\draw[color=red,line width=1.5pt] (8,25.5) -- (16,25.5);
\fill[color=red] (11,25.1) -- (15,25.1) -- (15,25.9) -- (11,25.9) -- cycle;
\draw[color=black,line width=1.5pt] (15,25.000000) -- (15,26.000000);
\draw[color=black,line width=1.5pt] (16,25.000000) -- (16,26.000000);
\draw[color=blue,line width=1.5pt] (30,24.5) -- (38,24.5);
\fill[color=blue] (36,24.1) -- (38,24.1) -- (38,24.9) -- (36,24.9) -- cycle;
\draw[color=black,line width=1.5pt] (32,24.000000) -- (32,25.000000);
\draw[color=black,line width=1.5pt] (38,24.000000) -- (38,25.000000);
\draw[color=black,dashed] (16,25.5)--(30,24.5);
\draw[color=red,line width=1.5pt] (52,25.5) -- (60,25.5);
\fill[color=red] (55,25.1) -- (59,25.1) -- (59,25.9) -- (55,25.9) -- cycle;
\draw[color=black,line width=1.5pt] (59,25.000000) -- (59,26.000000);
\draw[color=black,line width=1.5pt] (60,25.000000) -- (60,26.000000);
\draw[color=blue,line width=1.5pt] (74,24.5) -- (82,24.5);
\fill[color=blue] (80,24.1) -- (82,24.1) -- (82,24.9) -- (80,24.9) -- cycle;
\draw[color=black,line width=1.5pt] (76,24.000000) -- (76,25.000000);
\draw[color=black,line width=1.5pt] (82,24.000000) -- (82,25.000000);
\draw[color=black,dashed] (60,25.5)--(74,24.5);

\draw[color=red] (8,22.5) -- (82,22.5);
\draw[color=blue] (8,21.5) -- (82,21.5);
\draw[color=red,line width=1.5pt] (17,22.5) -- (23,22.5);
\fill[color=red] (17,22.1) -- (21,22.1) -- (21,22.9) -- (17,22.9) -- cycle;
\draw[color=black,line width=1.5pt] (22,22.000000) -- (22,23.000000);
\draw[color=black,line width=1.5pt] (23,22.000000) -- (23,23.000000);
\draw[color=blue,line width=1.5pt] (17,21.5) -- (22,21.5);
\fill[color=blue] (17,21.1) -- (19,21.1) -- (19,21.9) -- (17,21.9) -- cycle;
\draw[color=black,line width=1.5pt] (22,21.000000) -- (22,22.000000);
\draw[color=black,dashed] (23,22.5)--(38,21.5);

\draw[color=blue,line width=1.5pt] (38,21.5) -- (45,21.5);
\fill[color=blue] (38,21.1) -- (40,21.1) -- (40,21.9) -- (38,21.9) -- cycle;
\draw[color=black,line width=1.5pt] (43,21.000000) -- (43,22.000000);
\draw[color=black,line width=1.5pt] (45,21.000000) -- (45,22.000000);
\draw[color=red,line width=1.5pt] (38,22.5) -- (45,22.5);
\fill[color=red] (38,22.1) -- (42,22.1) -- (42,22.9) -- (38,22.9) -- cycle;
\draw[color=black,line width=1.5pt] (43,22.000000) -- (43,23.000000);
\draw[color=black,line width=1.5pt] (45,22.000000) -- (45,23.000000);
\draw[color=black,dashed] (45,22.5)--(61,21.5);

\draw[color=red,line width=1.5pt] (61,22.5) -- (67,22.5);
\fill[color=red] (61,22.1) -- (65,22.1) -- (65,22.9) -- (61,22.9) -- cycle;
\draw[color=black,line width=1.5pt] (66,22.000000) -- (66,23.000000);
\draw[color=black,line width=1.5pt] (67,22.000000) -- (67,23.000000);
\draw[color=blue,line width=1.5pt] (61,21.5) -- (67,21.5);
\fill[color=blue] (61,21.1) -- (63,21.1) -- (63,21.9) -- (61,21.9) -- cycle;
\draw[color=black,line width=1.5pt] (66,21.000000) -- (66,22.000000);
\draw[color=black,line width=1.5pt] (67,21.000000) -- (67,22.000000);
\draw[color=black,dashed] (67,22.5)--(82,21.5);

\draw[color=blue,line width=1.5pt] (82,21.5) -- (89,21.5);
\fill[color=blue] (82,21.1) -- (84,21.1) -- (84,21.9) -- (82,21.9) -- cycle;
\draw[color=black,line width=1.5pt] (87,21.000000) -- (87,22.000000);
\draw[color=black,line width=1.5pt] (89,21.000000) -- (89,22.000000);
\draw[color=red,line width=1.5pt] (82,22.5) -- (87,22.5);
\fill[color=red] (82,22.1) -- (86,22.1) -- (86,22.9) -- (82,22.9) -- cycle;
\draw[color=black,line width=1.5pt] (87,22.000000) -- (87,23.000000);

\draw[color=red] (8,19.5) -- (88,19.5);
\draw[color=blue] (8,18.5) -- (88,18.5);
\draw[color=blue,line width=1.5pt] (23,18.5) -- (27,18.5);
\fill[color=blue] (23,18.1) -- (25,18.1) -- (25,18.9) -- (23,18.9) -- cycle;
\draw[color=black,line width=1.5pt] (27,18.000000) -- (27,19.000000);
\draw[color=red,line width=1.5pt] (23,19.5) -- (29,19.5);
\fill[color=red] (23,19.1) -- (27,19.1) -- (27,19.9) -- (23,19.9) -- cycle;
\draw[color=black,line width=1.5pt] (27,19.000000) -- (27,20.000000);
\draw[color=black,line width=1.5pt] (29,19.000000) -- (29,20.000000);
\draw[color=black,dashed] (29,19.5)--(44,18.5);

\draw[color=blue,line width=1.5pt] (44,18.5) -- (51,18.5);
\fill[color=blue] (44,18.1) -- (46,18.1) -- (46,18.9) -- (44,18.9) -- cycle;
\draw[color=black,line width=1.5pt] (50,18.000000) -- (50,19.000000);
\draw[color=black,line width=1.5pt] (51,18.000000) -- (51,19.000000);
\draw[color=red,line width=1.5pt] (44,19.5) -- (51,19.5);
\fill[color=red] (44,19.1) -- (48,19.1) -- (48,19.9) -- (44,19.9) -- cycle;
\draw[color=black,line width=1.5pt] (50,19.000000) -- (50,20.000000);
\draw[color=black,line width=1.5pt] (51,19.000000) -- (51,20.000000);
\draw[color=black,dashed] (51,19.5)--(67,18.5);

\draw[color=blue,line width=1.5pt] (67,18.5) -- (73,18.5);
\fill[color=blue] (67,18.1) -- (69,18.1) -- (69,18.9) -- (67,18.9) -- cycle;
\draw[color=black,line width=1.5pt] (72,18.000000) -- (72,19.000000);
\draw[color=black,line width=1.5pt] (73,18.000000) -- (73,19.000000);
\draw[color=red,line width=1.5pt] (67,19.5) -- (73,19.5);
\fill[color=red] (67,19.1) -- (71,19.1) -- (71,19.9) -- (67,19.9) -- cycle;
\draw[color=black,line width=1.5pt] (72,19.000000) -- (72,20.000000);
\draw[color=black,line width=1.5pt] (73,19.000000) -- (73,20.000000);
\draw[color=black,dashed] (73,19.5)--(88,18.5);

\draw[color=blue,line width=1.5pt] (88,18.5) -- (95,18.5);
\fill[color=blue] (88,18.1) -- (90,18.1) -- (90,18.9) -- (88,18.9) -- cycle;
\draw[color=black,line width=1.5pt] (93,18.000000) -- (93,19.000000);
\draw[color=black,line width=1.5pt] (95,18.000000) -- (95,19.000000);
\draw[color=red,line width=1.5pt] (88,19.5) -- (93,19.5);
\fill[color=red] (88,19.1) -- (92,19.1) -- (92,19.9) -- (88,19.9) -- cycle;
\draw[color=black,line width=1.5pt] (93,19.000000) -- (93,20.000000);

\fill[color=black] (29,18) -- (29,28) -- (30,28) -- (30,18) -- cycle;
\fill[color=black] (51,18) -- (51,28) -- (52,28) -- (52,18) -- cycle;
\fill[color=black] (73,18) -- (73,28) -- (74,28) -- (74,18) -- cycle;
\fill[color=black] (95,18) -- (95,28) -- (96,28) -- (96,18) -- cycle;

\end{tikzpicture}
}\caption{\label{fig:reductie-assignment-allebeifalse}
For the same auxiliary scheduling problem instance as in Figure~\ref{fig:helereductie}, this partial schedule shows why setting $x_1=\textsc{false}$ and $x_2=\textsc{false}$ cannot lead to a feasible schedule.
In this partial schedule, we have fixed the jobs in each $\Lit{l}$ accordingly (note that the section where the literal is set to true has no idle time), and have marked in gray the space that is consumed by the jobs from these $\Lit{l}$ blocks.
The earliest starting times of the remaining jobs have been adjusted to this partial schedule.
This partial schedule cannot be extended to a feasible schedule because the jobs in the clause blocks~$\Cl{l}{1}$ of clause~$C_1 = (x_1 \vee x_2)$ cannot be scheduled without overlap.
}
\end{figure*}
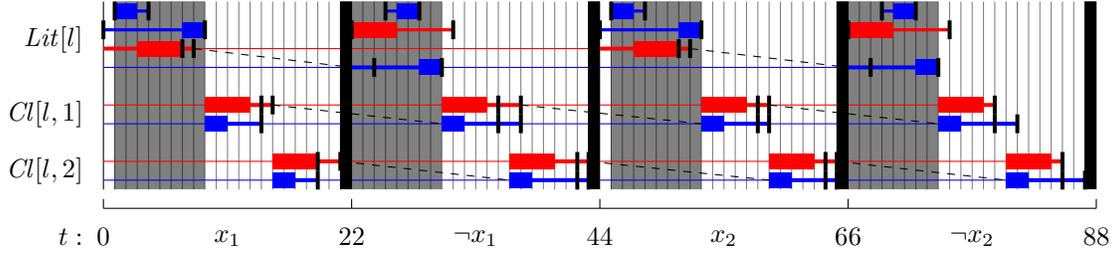

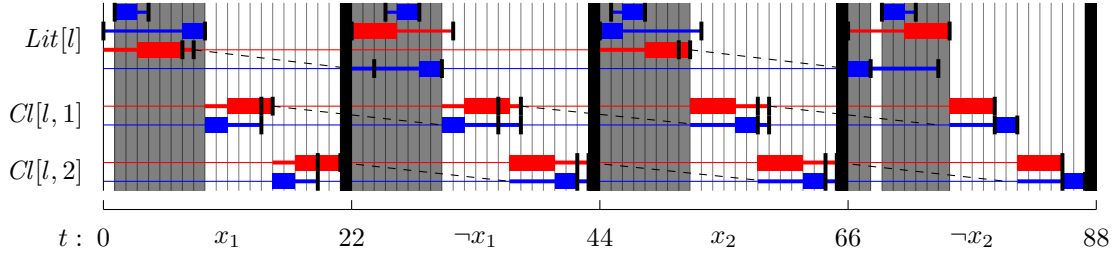
\begin{figure*}\centering
\makebox[\textwidth][c]{
\begin{tikzpicture}[xscale=0.15,yscale=0.25]
\draw[color=black] (8,18.0) -- (8,28);
\draw[color=gray] (9,18.0) -- (9,28);
\draw[color=gray] (10,18.0) -- (10,28);
\draw[color=gray] (11,18.0) -- (11,28);
\draw[color=gray] (12,18.0) -- (12,28);
\draw[color=gray] (13,18.0) -- (13,28);
\draw[color=gray] (14,18.0) -- (14,28);
\draw[color=gray] (15,18.0) -- (15,28);
\draw[color=gray] (16,18.0) -- (16,28);
\draw[color=gray] (17,18.0) -- (17,28);
\draw[color=gray] (18,18.0) -- (18,28);
\draw[color=gray] (19,18.0) -- (19,28);
\draw[color=gray] (20,18.0) -- (20,28);
\draw[color=gray] (21,18.0) -- (21,28);
\draw[color=gray] (22,18.0) -- (22,28);
\draw[color=gray] (23,18.0) -- (23,28);
\draw[color=gray] (24,18.0) -- (24,28);
\draw[color=gray] (25,18.0) -- (25,28);
\draw[color=gray] (26,18.0) -- (26,28);
\draw[color=gray] (27,18.0) -- (27,28);
\draw[color=gray] (28,18.0) -- (28,28);
\draw[color=gray] (29,18.0) -- (29,28);
\draw[color=gray] (30,18.0) -- (30,28);
\draw[color=gray] (31,18.0) -- (31,28);
\draw[color=gray] (32,18.0) -- (32,28);
\draw[color=gray] (33,18.0) -- (33,28);
\draw[color=gray] (34,18.0) -- (34,28);
\draw[color=gray] (35,18.0) -- (35,28);
\draw[color=gray] (36,18.0) -- (36,28);
\draw[color=gray] (37,18.0) -- (37,28);
\draw[color=gray] (38,18.0) -- (38,28);
\draw[color=gray] (39,18.0) -- (39,28);
\draw[color=gray] (40,18.0) -- (40,28);
\draw[color=gray] (41,18.0) -- (41,28);
\draw[color=gray] (42,18.0) -- (42,28);
\draw[color=gray] (43,18.0) -- (43,28);
\draw[color=gray] (44,18.0) -- (44,28);
\draw[color=gray] (45,18.0) -- (45,28);
\draw[color=gray] (46,18.0) -- (46,28);
\draw[color=gray] (47,18.0) -- (47,28);
\draw[color=gray] (48,18.0) -- (48,28);
\draw[color=gray] (49,18.0) -- (49,28);
\draw[color=gray] (50,18.0) -- (50,28);
\draw[color=gray] (51,18.0) -- (51,28);
\draw[color=gray] (52,18.0) -- (52,28);
\draw[color=gray] (53,18.0) -- (53,28);
\draw[color=gray] (54,18.0) -- (54,28);
\draw[color=gray] (55,18.0) -- (55,28);
\draw[color=gray] (56,18.0) -- (56,28);
\draw[color=gray] (57,18.0) -- (57,28);
\draw[color=gray] (58,18.0) -- (58,28);
\draw[color=gray] (59,18.0) -- (59,28);
\draw[color=gray] (60,18.0) -- (60,28);
\draw[color=gray] (61,18.0) -- (61,28);
\draw[color=gray] (62,18.0) -- (62,28);
\draw[color=gray] (63,18.0) -- (63,28);
\draw[color=gray] (64,18.0) -- (64,28);
\draw[color=gray] (65,18.0) -- (65,28);
\draw[color=gray] (66,18.0) -- (66,28);
\draw[color=gray] (67,18.0) -- (67,28);
\draw[color=gray] (68,18.0) -- (68,28);
\draw[color=gray] (69,18.0) -- (69,28);
\draw[color=gray] (70,18.0) -- (70,28);
\draw[color=gray] (71,18.0) -- (71,28);
\draw[color=gray] (72,18.0) -- (72,28);
\draw[color=gray] (73,18.0) -- (73,28);
\draw[color=gray] (74,18.0) -- (74,28);
\draw[color=gray] (75,18.0) -- (75,28);
\draw[color=gray] (76,18.0) -- (76,28);
\draw[color=gray] (77,18.0) -- (77,28);
\draw[color=gray] (78,18.0) -- (78,28);
\draw[color=gray] (79,18.0) -- (79,28);
\draw[color=gray] (80,18.0) -- (80,28);
\draw[color=gray] (81,18.0) -- (81,28);
\draw[color=gray] (82,18.0) -- (82,28);
\draw[color=gray] (83,18.0) -- (83,28);
\draw[color=gray] (84,18.0) -- (84,28);
\draw[color=gray] (85,18.0) -- (85,28);
\draw[color=gray] (86,18.0) -- (86,28);
\draw[color=gray] (87,18.0) -- (87,28);
\draw[color=gray] (88,18.0) -- (88,28);
\draw[color=gray] (89,18.0) -- (89,28);
\draw[color=gray] (90,18.0) -- (90,28);
\draw[color=gray] (91,18.0) -- (91,28);
\draw[color=gray] (92,18.0) -- (92,28);
\draw[color=gray] (93,18.0) -- (93,28);
\draw[color=gray] (94,18.0) -- (94,28);
\draw[color=gray] (95,18.0) -- (95,28);
\draw[color=gray] (96,18.0) -- (96,28);

\draw[color=black] (8.000000,17) -- (30.000000,17);
\draw[color=black] (8.000000,17) -- (8.000000,17.5);
\node[draw=none,fill=none,anchor=mid] at (19,15.5) {$x_1$};
\draw[color=black] (30.000000,17) -- (52.000000,17);
\draw[color=black] (30.000000,17) -- (30.000000,17.5);
\node[draw=none,fill=none,anchor=mid] at (41,15.5) {$\neg x_1$};
\draw[color=black] (52.000000,17) -- (74.000000,17);
\draw[color=black] (52.000000,17) -- (52.000000,17.5);
\node[draw=none,fill=none,anchor=mid] at (63,15.5) {$x_2$};
\draw[color=black] (74.000000,17) -- (96.000000,17);
\draw[color=black] (74.000000,17) -- (74.000000,17.5);
\node[draw=none,fill=none,anchor=mid] at (85,15.5) {$\neg x_2$};
\draw[color=black] (96.000000,17) -- (96.000000,17.5);
\node[draw=none,fill=none,anchor=east] at (7,15.5) {  $t:$};
\node[draw=none,fill=none,anchor=east] at (7,19) {  $\Cl{l}{2}$};
\node[draw=none,fill=none,anchor=east] at (7,22) {  $\Cl{l}{1}$};
\node[draw=none,fill=none,anchor=east] at (7,26) {  $\Lit{l}$};
\node[draw=none,fill=none] at (8,15.5) {  $0$};
\node[draw=none,fill=none] at (30,15.5) {  $22$};
\node[draw=none,fill=none] at (52,15.5) {  $44$};
\node[draw=none,fill=none] at (74,15.5) {  $66$};
\node[draw=none,fill=none] at (96,15.5) {  $88$};

\fill[opacity=0.5] (9,28) -- (17,28) -- (17,18) -- (9,18) -- cycle;
\fill[opacity=0.5] (30,28) -- (38,28) -- (38,18) -- (30,18) -- cycle;
\fill[opacity=0.5] (52,28) -- (60,28) -- (60,18) -- (52,18) -- cycle;
\fill[opacity=0.5] (74,28) -- (76,28) -- (76,18) -- (74,18) -- cycle;
\fill[opacity=0.5] (77,28) -- (83,28) -- (83,18) -- (77,18) -- cycle;

\draw[color=blue,very thick] (9,27.5) -- (12,27.5);
\fill[color=blue] (9,27.1) -- (11,27.1) -- (11,27.9) -- (9,27.9) -- cycle;
\draw[color=black,line width=1.5pt] (9,27) -- (9,28);
\draw[color=black,line width=1.5pt] (12,27) -- (12,28);
\draw[color=blue,very thick] (8,26.5) -- (17,26.5);
\fill[color=blue] (15,26.1) -- (17,26.1) -- (17,26.9) -- (15,26.9) -- cycle;
\draw[color=black,line width=1.5pt] (8,26) -- (8,27);
\draw[color=black,line width=1.5pt] (17,26) -- (17,27);
\draw[color=blue,very thick] (33,27.5) -- (36,27.5);
\fill[color=blue] (34,27.1) -- (36,27.1) -- (36,27.9) -- (34,27.9) -- cycle;
\draw[color=black,line width=1.5pt] (33,27) -- (33,28);
\draw[color=black,line width=1.5pt] (36,27) -- (36,28);
\draw[color=red,very thick] (30,26.5) -- (39,26.5);
\fill[color=red] (30,26.1) -- (34,26.1) -- (34,26.9) -- (30,26.9) -- cycle;
\draw[color=black,line width=1.5pt] (30,26) -- (30,27);
\draw[color=black,line width=1.5pt] (39,26) -- (39,27);
\draw[color=blue,very thick] (53,27.5) -- (56,27.5);
\fill[color=blue] (54,27.1) -- (56,27.1) -- (56,27.9) -- (54,27.9) -- cycle;
\draw[color=black,line width=1.5pt] (53,27) -- (53,28);
\draw[color=black,line width=1.5pt] (56,27) -- (56,28);
\draw[color=blue,very thick] (52,26.5) -- (61,26.5);
\fill[color=blue] (52,26.1) -- (54,26.1) -- (54,26.9) -- (52,26.9) -- cycle;
\draw[color=black,line width=1.5pt] (52,26) -- (52,27);
\draw[color=black,line width=1.5pt] (61,26) -- (61,27);
\draw[color=blue,very thick] (77,27.5) -- (80,27.5);
\fill[color=blue] (77,27.1) -- (79,27.1) -- (79,27.9) -- (77,27.9) -- cycle;
\draw[color=black,line width=1.5pt] (77,27) -- (77,28);
\draw[color=black,line width=1.5pt] (80,27) -- (80,28);
\draw[color=red,very thick] (74,26.5) -- (83,26.5);
\fill[color=red] (79,26.1) -- (83,26.1) -- (83,26.9) -- (79,26.9) -- cycle;
\draw[color=black,line width=1.5pt] (74,26) -- (74,27);
\draw[color=black,line width=1.5pt] (83,26) -- (83,27);

\draw[color=red] (8,25.5) -- (52,25.5);
\draw[color=blue] (8,24.5) -- (74,24.5);
\draw[color=red,line width=1.5pt] (8,25.5) -- (16,25.5);
\fill[color=red] (11,25.1) -- (15,25.1) -- (15,25.9) -- (11,25.9) -- cycle;
\draw[color=black,line width=1.5pt] (15,25.000000) -- (15,26.000000);
\draw[color=black,line width=1.5pt] (16,25.000000) -- (16,26.000000);
\draw[color=blue,line width=1.5pt] (30,24.5) -- (38,24.5);
\fill[color=blue] (36,24.1) -- (38,24.1) -- (38,24.9) -- (36,24.9) -- cycle;
\draw[color=black,line width=1.5pt] (32,24.000000) -- (32,25.000000);
\draw[color=black,line width=1.5pt] (38,24.000000) -- (38,25.000000);
\draw[color=black,dashed] (16,25.5)--(30,24.5);
\draw[color=red,line width=1.5pt] (52,25.5) -- (60,25.5);
\fill[color=red] (56,25.1) -- (60,25.1) -- (60,25.9) -- (56,25.9) -- cycle;
\draw[color=black,line width=1.5pt] (59,25.000000) -- (59,26.000000);
\draw[color=black,line width=1.5pt] (60,25.000000) -- (60,26.000000);
\draw[color=blue,line width=1.5pt] (74,24.5) -- (82,24.5);
\fill[color=blue] (74,24.1) -- (76,24.1) -- (76,24.9) -- (74,24.9) -- cycle;
\draw[color=black,line width=1.5pt] (76,24.000000) -- (76,25.000000);
\draw[color=black,line width=1.5pt] (82,24.000000) -- (82,25.000000);
\draw[color=black,dashed] (60,25.5)--(74,24.5);

\draw[color=red] (8,22.5) -- (83,22.5);
\draw[color=blue] (8,21.5) -- (83,21.5);
\draw[color=red,line width=1.5pt] (17,22.5) -- (23,22.5);
\fill[color=red] (19,22.1) -- (23,22.1) -- (23,22.9) -- (19,22.9) -- cycle;
\draw[color=black,line width=1.5pt] (22,22.000000) -- (22,23.000000);
\draw[color=black,line width=1.5pt] (23,22.000000) -- (23,23.000000);
\draw[color=blue,line width=1.5pt] (17,21.5) -- (22,21.5);
\fill[color=blue] (17,21.1) -- (19,21.1) -- (19,21.9) -- (17,21.9) -- cycle;
\draw[color=black,line width=1.5pt] (22,21.000000) -- (22,22.000000);
\draw[color=black,dashed] (23,22.5)--(38,21.5);

\draw[color=blue,line width=1.5pt] (38,21.5) -- (45,21.5);
\fill[color=blue] (38,21.1) -- (40,21.1) -- (40,21.9) -- (38,21.9) -- cycle;
\draw[color=black,line width=1.5pt] (43,21.000000) -- (43,22.000000);
\draw[color=black,line width=1.5pt] (45,21.000000) -- (45,22.000000);
\draw[color=red,line width=1.5pt] (38,22.5) -- (45,22.5);
\fill[color=red] (40,22.1) -- (44,22.1) -- (44,22.9) -- (40,22.9) -- cycle;
\draw[color=black,line width=1.5pt] (43,22.000000) -- (43,23.000000);
\draw[color=black,line width=1.5pt] (45,22.000000) -- (45,23.000000);
\draw[color=black,dashed] (45,22.5)--(61,21.5);

\draw[color=red,line width=1.5pt] (60,22.5) -- (67,22.5);
\fill[color=red] (60,22.1) -- (64,22.1) -- (64,22.9) -- (60,22.9) -- cycle;
\draw[color=black,line width=1.5pt] (66,22.000000) -- (66,23.000000);
\draw[color=black,line width=1.5pt] (67,22.000000) -- (67,23.000000);
\draw[color=blue,line width=1.5pt] (60,21.5) -- (67,21.5);
\fill[color=blue] (64,21.1) -- (66,21.1) -- (66,21.9) -- (64,21.9) -- cycle;
\draw[color=black,line width=1.5pt] (66,21.000000) -- (66,22.000000);
\draw[color=black,line width=1.5pt] (67,21.000000) -- (67,22.000000);
\draw[color=black,dashed] (67,22.5)--(82,21.5);

\draw[color=blue,line width=1.5pt] (83,21.5) -- (89,21.5);
\fill[color=blue] (87,21.1) -- (89,21.1) -- (89,21.9) -- (87,21.9) -- cycle;
\draw[color=black,line width=1.5pt] (87,21.000000) -- (87,22.000000);
\draw[color=black,line width=1.5pt] (89,21.000000) -- (89,22.000000);
\draw[color=red,line width=1.5pt] (83,22.5) -- (87,22.5);
\fill[color=red] (83,22.1) -- (87,22.1) -- (87,22.9) -- (83,22.9) -- cycle;
\draw[color=black,line width=1.5pt] (87,22.000000) -- (87,23.000000);

\draw[color=red] (8,19.5) -- (89,19.5);
\draw[color=blue] (8,18.5) -- (89,18.5);
\draw[color=blue,line width=1.5pt] (23,18.5) -- (27,18.5);
\fill[color=blue] (23,18.1) -- (25,18.1) -- (25,18.9) -- (23,18.9) -- cycle;
\draw[color=black,line width=1.5pt] (27,18.000000) -- (27,19.000000);
\draw[color=red,line width=1.5pt] (23,19.5) -- (29,19.5);
\fill[color=red] (25,19.1) -- (29,19.1) -- (29,19.9) -- (25,19.9) -- cycle;
\draw[color=black,line width=1.5pt] (27,19.000000) -- (27,20.000000);
\draw[color=black,line width=1.5pt] (29,19.000000) -- (29,20.000000);
\draw[color=black,dashed] (29,19.5)--(44,18.5);

\draw[color=blue,line width=1.5pt] (44,18.5) -- (51,18.5);
\fill[color=blue] (48,18.1) -- (50,18.1) -- (50,18.9) -- (48,18.9) -- cycle;
\draw[color=black,line width=1.5pt] (50,18.000000) -- (50,19.000000);
\draw[color=black,line width=1.5pt] (51,18.000000) -- (51,19.000000);
\draw[color=red,line width=1.5pt] (44,19.5) -- (51,19.5);
\fill[color=red] (44,19.1) -- (48,19.1) -- (48,19.9) -- (44,19.9) -- cycle;
\draw[color=black,line width=1.5pt] (50,19.000000) -- (50,20.000000);
\draw[color=black,line width=1.5pt] (51,19.000000) -- (51,20.000000);
\draw[color=black,dashed] (51,19.5)--(67,18.5);

\draw[color=blue,line width=1.5pt] (66,18.5) -- (73,18.5);
\fill[color=blue] (70,18.1) -- (72,18.1) -- (72,18.9) -- (70,18.9) -- cycle;
\draw[color=black,line width=1.5pt] (72,18.000000) -- (72,19.000000);
\draw[color=black,line width=1.5pt] (73,18.000000) -- (73,19.000000);
\draw[color=red,line width=1.5pt] (66,19.5) -- (73,19.5);
\fill[color=red] (66,19.1) -- (70,19.1) -- (70,19.9) -- (66,19.9) -- cycle;
\draw[color=black,line width=1.5pt] (72,19.000000) -- (72,20.000000);
\draw[color=black,line width=1.5pt] (73,19.000000) -- (73,20.000000);
\draw[color=black,dashed] (73,19.5)--(88,18.5);

\draw[color=blue,line width=1.5pt] (89,18.5) -- (95,18.5);
\fill[color=blue] (93,18.1) -- (95,18.1) -- (95,18.9) -- (93,18.9) -- cycle;
\draw[color=black,line width=1.5pt] (93,18.000000) -- (93,19.000000);
\draw[color=black,line width=1.5pt] (95,18.000000) -- (95,19.000000);
\draw[color=red,line width=1.5pt] (89,19.5) -- (93,19.5);
\fill[color=red] (89,19.1) -- (93,19.1) -- (93,19.9) -- (89,19.9) -- cycle;
\draw[color=black,line width=1.5pt] (93,19.000000) -- (93,20.000000);

\fill[color=black] (29,18) -- (29,28) -- (30,28) -- (30,18) -- cycle;
\fill[color=black] (51,18) -- (51,28) -- (52,28) -- (52,18) -- cycle;
\fill[color=black] (73,18) -- (73,28) -- (74,28) -- (74,18) -- cycle;
\fill[color=black] (95,18) -- (95,28) -- (96,28) -- (96,18) -- cycle;

\end{tikzpicture}
}\caption{\label{fig:reductie-assignment-x1false-x2true}
For the same auxiliary scheduling problem instance as in Figure~\ref{fig:helereductie}, this shows a feasible schedule corresponding to setting $x_1=\textsc{false}$ and $x_2=\textsc{true}$.
We have marked in gray the space that is consumed by the jobs from the $\Lit{l}$ blocks.
}
\end{figure*}

\paragraph{Consistency of literals}
We now study the example in more detail.
The first observation, which is formalized in Proposition~\ref{prop:thegadgets}, is that the $\Lit{l}$ blocks have two possible feasible schedules, of which one when its pending job must meet the early deadline,
but this comes at the cost of introducing one unit of idle time which cannot be filled by other jobs.
Because $\Lit{x_i}$ and $\Lit{\neg x_i}$ are connected via the pending job pair $v_i$, one of the sections of $x_i$ and $\neg x_i$ contains a unit of idle time at the start (in any feasible schedule). This corresponds to setting $x_i$ to true or false (the section without idle time is the literal which is set to true). 
In the example in Figure~\ref{fig:helereductie} this can be verified by considering the four first lines and the two left-most sections, containing $V^+$ and $V^-$ for $x_1$ and $\neg x_1$, respectively, and similarly for $x_2$ in the third and fourth section.

\paragraph{Definition of clause blocks}
For $\Cactive$ blocks, the early deadline of 6 allows both pending jobs of length 2 and 4 to finish early if there is no idle time in the $\Lit{l}$ block;
if there is idle time, one job can complete early, and the other late.
For $\Cinactive$ blocks, the early deadline of 5 allows one job to finish early even with one unit of idle time, and the other will be late.
As noted before, the $\Cl{l}{j}$ are defined such that they are sensitive to the unit of delay at the start only if $l \in C_j$. In other words, the $\Cl{l}{j}$ are sensitive to whether $l$ is set to true only if $l \in C_j$, and thus $\Cl{l}{j}$ is of type $\Cactive$.

\paragraph{Verifying the encoding of clauses by the clause blocks}
We now look at how the pending jobs for clause $C_1 = (x_1 \vee x_2)$ encode that $C_1$ must be satisfied.
We are going to schedule the $\Cl{l}{1}$ blocks.
We say that we schedule in the order \PQ (respectively \QP) if we put the long job (respectively the short job) first in $\Cl{l}{1}$.

To show that an unsatisfying assignment cannot result in a feasible schedule, consider the assignment $x_1=x_2=\textsc{false}$, corresponding to scheduling the pending jobs in~$\Lit{x_1}$ and~$\Lit{x_2}$ early.
The scheduling instance corresponding to this assignment is shown in Figure~\ref{fig:reductie-assignment-allebeifalse}.
Here we have already fixed a schedule to the~$\Lit{l}$ blocks.
We schedule the sections from left to right.
For $x_1$, $\Cl{x_1}{1} = \Cactive$ and the relative deadlines are 6 and 7. Because~$\Lit{x_1}$ completes one unit of time late, one job can complete early, and this must be the job with only an early deadline.
Therefore we put the jobs in the order \QP and we have to schedule the short job in~$\Cl{\neg x_1}{1}$ early.
For $\neg x_1$, $\Cl{\neg x_1}{1} = \Cinactive$ and the relative early deadline is 5, therefore we must put the jobs in the order \QP and we have to schedule the short job in~$\Cl{x_2}{1}$ early.
For $x_2$, $\Cl{x_2}{1} = \Cactive$ and the relative early deadline is 6, but we must take into account the unit of delay from the $\Lit{x_2}$ block's schedule, so again we must put the jobs in the order \QP
and the long job will be scheduled late, so the short job in $\Cl{\neg x_2}{1}$ must be scheduled early.
However in the section of $\neg x_2$, $\Cl{\neg x_2}{1} = \Cinactive$ and both the short pending job and the unpaired long job must be scheduled early, which is impossible.

This argument shows that if all $\Cactive$ blocks of a clause are in sections where the pending job in $\Lit{l}$ completes early, then in all sections we must schedule the $\Cl{l}{j}$ block in the order \QP, but this doesn't work for the last section, so that this condition means that a schedule does not exist at all.

If instead we consider $x_1=\textsc{false},x_2=\textsc{true}$, that is, having the idle time in the first and fourth section,
then it can be verified that scheduling the $\Cl{l}{1}$ blocks from left to right in the order (\QP, \QP, \PQ, \PQ) or (\QP, \QP, \QP, \PQ) leads to a feasible schedule (see Figure~\ref{fig:reductie-assignment-x1false-x2true} for the former).


In general, the pending jobs of a clause $C_j$ can be feasibly scheduled if an $\Cactive$ block in a section for literal $l\in C_j$ has the related pending job in $\Lit{l}$ finished after the early deadline.
This is formalized by Proposition~\ref{prop:blokdetail} in the next section.
A formula then is satisfiable if each clause can be feasibly scheduled in this manner.

\subsection{General reduction to $\operatorname{AUX}(p,q)$}
In this section we describe the reduction more formally, but follow the same steps.
Let a propositional formula be given in conjunctive normal form with $n$ variables and $m$ clauses.

At a high level, the instance of the $\operatorname{AUX}(p,q)$ scheduling problem consists of $2n$ sections, one per literal.
Each section has the same length $S$.
In each section we define a set of jobs with deadlines in $[0, S]$.
Then, we concatenate the sections by shifting the release times and deadlines.
The total sum of job lengths in each section is $S-1$, and the last time units $[S-q,S)$ are occupied by a separator job.
Before the separator, sections contain $m+1$ \emph{blocks} of jobs. The first block represents the literal and the others the $m$ clauses: each clause is represented in each literal's section, regardless of whether the clause contains the literal.
For literal $l$, we denote these blocks by $\Lit{l}$ and $\Cl{l}{j}$.
The four types of blocks we use are defined as follows.

\begin{definition}[Blocks]\label{def:blocks-def}
We use four types of blocks, $V^+$, $V^-$, $\Cactive$ and $\Cinactive$:
\begin{enumerate}
\item
The $V^+$ block, used as $\Lit{l}$ for positive literals $x_i$,
defines a long pending job with deadlines~$(d'_p,d_p)=(p+q+1,p+2q)$ and two ordinary jobs, $([1,2q],q)$ and $([0,p+2q+1],q)$.
\item
The $V^-$ block, used as $\Lit{l}$ for negative literals $\neg{x}_i$,
defines a short pending job with two deadlines~$(d'_q,d_q)=(q,p+2q)$ and two ordinary jobs, $([q+1,p+q],q)$ and $([0,p+2q+1],p)$.
\item The $\Cactive$ block, used to represent clauses containing a literal,
defines two pending jobs with deadlines $(d'_p,d_p)=(d'_q,d_q)=(p+q,p+q+1)$.
\item The $\Cinactive$ block, used to represent clauses not containing a literal,
defines two pending jobs with deadlines $(d'_p,d_p)=(d'_q,d_q)=(p+q-1,p+q+1)$.
\end{enumerate}
\end{definition}
These blocks are illustrated by Figure~\ref{fig:hetplaatje}.

When we say that we \emph{shift} a set of jobs~$J$ by a time offset $\Delta t$, we mean that we add $\Delta t$ to the release times and deadlines of ordinary jobs in~$J$, and to the deadlines of pending jobs in~$J$.
The reduction is defined as follows.
\begin{definition}[Reduction]\label{lbl:reductie-def}
Given an instance of the Satisfiability problem with $n$ variables $x_1,x_2,\ldots,x_n$ and $m$ clauses $C_1, C_2, \ldots, C_m$,
we construct an instance of $\operatorname{AUX}(p,q)$ with $n + 2nm$ jobs for each of the two different job lengths using the blocks as follows.
For each literal~$l$, we define a section, consisting of $\Lit{l}=V^+$ for positive literals $l=x_i$ and $\Lit{l}=V^-$ for negative literals $l=\neg{x}_i$, and $\Cl{l}{j}=\Cactive$ if $l \in C_j$ and $\Cl{l}{j}=\Cinactive$ otherwise.
Let~$S = p+2q+m(p+q)+1+q$.
The section for literal~$l$ consists of the following jobs:
\begin{itemize}
\item $\Lit{l}$.
\item For $j=1,\ldots,m$, $\Cl{l}{j}$ shifted by $p+2q+(j-1)(p+q)$ time units.
\item A separator job occupying $[S-q,S)$.
\end{itemize}
The complete instance is the concatenation of these sections for all literals in the order $x_1,\neg{x}_1, \ldots, x_n,\neg{x}_n$,
where each section is shifted $S$ time units further than its predecessor.

For each variable $x_i$, the long pending job in $\Lit{x_i}=V^+$ is paired with the short pending job in $\Lit{\neg{x_i}}=V^-$.
For each $i=1,\ldots,n$, for each clause $j=1,\ldots,m$, the long pending job in $\Cl{x_i}{j}$ is paired with the short pending job in $\Cl{\neg{x_i}}{j}$, and the long pending job in $\Cl{\neg{x_i}}{j}$ is paired with the short pending job in $\Cl{x_{i+1}}{j}$, except for $i=n$.
These last $m$ long pending jobs as well as the first $m$ short pending jobs remain unpaired, and, additionally, are required to complete by their early deadline, effectively being ordinary jobs, and thus leaving $N = n + (2n-1)m$ pairs of connected pending jobs.
\end{definition}

The above definition follows the earlier discussion on the literal and clause blocks.
An instance of $\operatorname{AUX}(p,q)$ as per Definition~\ref{defn:model} can be obtained by separating the long pending jobs, the short pending jobs, and the ordinary jobs into the sets $J_p$, $J_q$ and $J$.
The resulting instance has $O(n m)$ jobs in total, which is polynomial in the size of the SAT formula, and which does not depend on $p$ or $q$.

The connected pairs of pending jobs between blocks are as follows.
The pending jobs in $\Lit{x_i}$ and $\Lit{\neg x_i}$ are connected for all $1\leq i\leq n$.
Further, for each $1\leq j\leq m$ and $1\leq k\leq 2n-1$, the long pending job in $\Cl{\idxtolit{k}}{j}$ and the short pending job in $\Cl{\idxtolit{k+1}}{j}$ are connected.

\subsection{Reduction proof}
To prove that any solution to satisfiability can be translated to a solution for the auxiliary scheduling problem, and vice versa, we first derive a few properties of the defined schedule.

\begin{proposition}\label{prop:thegadgets}
For~$\Lit{l}$, two schedules are possible:
one without idle time and completion time $p+2q$,
and one with one unit of idle time and completion time $p+2q+1$;
only in the latter schedule the pending job can complete early.
\end{proposition}
\begin{proof}
We consider both cases~$\Lit{l}=V^+$ and $\Lit{l}=V^-$.
The blocks are also visualized in Figure~\ref{fig:hetplaatje}.

For~$V^+$:
because of the short job with interval $[1,2q]$, we must start with a short job.
If we start with the $[1,2q]$ job, it must start at time 1 because
the last deadline is~$p+2q+1$ and we already have one idle unit.
So the long pending job is scheduled at $[q+1,p+q+1)$ and the other short job at $[p+q+1,p+2q+1)$.
Otherwise, two short jobs must occupy $[0,2q)$ and the pending job occupies $[2q,p+2q)$.
Only the first schedule meets the pending job's early deadline.

For~$V^-$:
if we schedule the short pending job first on~$[0,q)$ or $[1,q+1)$,
the other two occupy $[q+1,p+2q+1)$.
Otherwise, we must start with the long job, and because of the other short job,
the two ordinary jobs occupy $[0,p+q)$, and the pending job occupies $[p+q,p+2q)$.
Again, only in the first schedule the pending job can meet its early deadline.
\end{proof}

\begin{lemma}[SAT model $\Rightarrow$ feasible schedule]\label{lem:dir1}
Let $v(x_1),v(x_2),\ldots,v(x_n)$ be a model for the SAT formula.
Then the constructed scheduling instance has a feasible schedule.
\end{lemma}
\begin{proof}
We specify the schedules for the sections and then concatenate these.
The idea is to schedule all jobs in the order of their blocks $\Lit{l}, \Cl{l}{1}, \ldots, \Cl{l}{m}$.
For all $1\leq i\leq n$, if $v(x_i)=\textsc{true}$, require the pending job in $\Lit{\neg x_i}$ to finish early;
if $v(x_i)=\textsc{false}$, require the pending job in $\Lit{x_i}$ to finish early.
By Proposition~\ref{prop:thegadgets}, this allows us to schedule $\Lit{l}$ for literals set to true without idle time.
For clauses containing such literals, $\Cl{l}{j} = \Cactive$ and we can schedule both pending jobs in $\Cl{l}{j}$ early.
For all other $\Cl{l}{j}$, we can schedule exactly one pending job early.
For each clause $C_j$, let $l=\idxtolit{k} \in C_j$ be a literal set to true by~$v$ and $k$ be its index.
For $k' < k$, schedule the short job before the long one; for $k' = k$ the order does not matter, and for $k' > k$ schedule the long job before the short one.
Observe that because $l$ is set to \textsc{true}, it is possible to schedule both jobs in $\Cl{l}{j}$ early.
So we schedule the short jobs in $\Cl{x_1}{j}, \ldots, \Cl{l}{j}$ and the long jobs in $\Cl{l}{j}, \ldots, \Cl{\neg x_n}{j}$ early.
This satisfies the conditions for the pending jobs.
\end{proof}

Now we derive a few properties of the constructed scheduling instance to prove the other direction.
\begin{proposition}\label{prop:blk}
In any feasible schedule, all jobs (except the separator jobs) must have a start time after the time offset to which their block was shifted.
\end{proposition}
\begin{proof}
Because the jobs in each section have length $S-1$, each block ends with a separator job, and all job lengths are greater than one,
it follows that all jobs must be scheduled within their own section.
For a given section of literal~$l$, the claim clearly holds for the initial block~$\Lit{l}$.
For the $\Cl{l}{j}$ blocks, note that $\Cl{l}{j}$ is shifted by $p+2q+(j-1)(p+q)$, which equals the total job length of $\Lit{l}$ and all $\Cl{l}{j'}$ with $j' < j$.
Because jobs in both $\Lit{l}$ and $\Cl{l}{j}$ have deadlines exceeding the block's total job length by at most one, and $p, q > 1$,
the jobs in $\Cl{l}{j}$ must be scheduled after those in $\Lit{l}$ and $\Cl{l}{j'}$ for $j' < j$,
so they can start at time equal to their offset or later.
\end{proof}

To proof that not satisfying a clause implies an infeasible schedule, we use the following fact on the relation between the pending job defined in a literal and the corresponding clause jobs.

\begin{proposition}\label{prop:blokdetail}
Consider any clause $C_j$. In any feasible schedule, there must exist a literal $l \in C_j$ such that the pending job in $\Lit{l}$ completes after its early deadline.
\end{proposition}
\begin{proof}
We argue by contradiction.
Suppose that for all literals~$l \in C_j$, the pending job in~$\Lit{l}$ completes early.
By Proposition~\ref{prop:blk}, the jobs of each block are scheduled no earlier than the offset of their block.
So we can apply Proposition~\ref{prop:thegadgets} to the $\Lit{l}$ blocks.
For literals~$l \in C_j$, the $\Lit{l}$ blocks are scheduled with one unit of idle time.
This unit of delay propagates through the $\Cl{l}{j}$ blocks.
So for $l \in C_j$, we can schedule the jobs in $\Cl{l}{j}$ starting 1 time unit later than the block's offset.
Because the early deadlines in $\Cactive$ are $p+q$, we have that for $l\in C_j$, one pending job in~$\Cl{l}{j}$ completes late.
For literals~$l \notin C_j$, we can schedule the jobs in $\Cl{l}{j}$ starting at the block's offset (or later), and the early deadlines are $p+q-1$, so also for these literals, one pending job in~$\Cl{l}{j}$ completes late.
Therefore, $2n$ jobs complete late among the jobs in $\Cl{x_1}{j}, \ldots, \Cl{\neg x_n}{j}$.
The unpaired pending jobs must be early, and there are $2n-1$ pending job pairs contained in these blocks, so the number of pending jobs that will be late ($2n$) exceeds the number of pending jobs that can be late ($2n-1$), which is a contradiction.
\end{proof}

\begin{lemma}[Feasible schedule $\Rightarrow$ SAT model]\label{lem:dir2}
Suppose we have a feasible schedule for the constructed scheduling instance.
Then there exists a model for the SAT formula.
\end{lemma}
\begin{proof}
For all $1\leq i\leq n$, if the pending job in~$\Lit{x_i}$ completes late, set $v(x_i)=\textsc{true}$;
if the pending job in~$\Lit{\neg x_i}$ completes late, set $v(x_i)=\textsc{false}$ (because the pending jobs are paired, not both can complete late).
If neither completes late, set $v(x_i)$ arbitrarily.

We show that this assignment satisfies all clauses.
Consider a clause $C_j$.
By Proposition~\ref{prop:blokdetail}, there exists a literal~$l \in C_j$ such that the pending job in~$\Lit{l}$ completes after its early deadline.
But then, this literal was set to true by the assignment we have defined, so the clause is satisfied.
\end{proof}

Since the reduction is polynomial, Lemmas~\ref{lem:dir1} and~\ref{lem:dir2} together prove the main Lemma~\ref{lem:partB}, that is, that $\operatorname{AUX}(p,q)$ is strongly NP-complete for $p > q > 1$.
Together with the reducibility of $\operatorname{AUX}(p,q)$ to our original problem (Lemma~\ref{lem:partA}), we establish the main result of this paper (Theorem~\ref{thm:mainresult}), which is that the original scheduling problem with job lengths $p > q > 1$ is strongly NP-complete.

\section{Discussion}
Proving NP-completeness of the two job lengths problem turned out to be much more complex than for the unrestricted problem where we can simply reduce from 3-Partition.
Since the reduction proof in this paper is valid and polynomial for constant length $p > q > 1$, it shows that the non-preemptive job scheduling problem with release times and deadlines and such job lengths $p$ and $q$ is \emph{strongly} NP-complete.
A direct consequence of our result is that the case~$\{1,p\}$ is the maximally theoretically solvable case.

Our result contradicts the existence of a pseudo-polynomial algorithm for the case $\{q,2q\}$ as described by \cite{DBLP:journals/anor/Vakhania04}.
One might think that when the long job length is a multiple of the short one, the problem can 
 be translated to the~$\{1,p\}$ problem, but this is not possible.
Intuitively, with two non-unit job lengths we are able to force the schedule to contain
idle time between the execution intervals of two consecutive jobs, while
for $\{1,p\}$ unit length jobs can always be inserted into idle intervals of unit length.

Another way to think about the~$\{1,p\}$ problem is that the unit length jobs are actually preemptive jobs:
the linear programming formulation of~\cite{DBLP:conf/esa/Sgall12} only works because the long jobs never need to preempt the unit length jobs.
Therefore, we can consider a generalized problem with both preemptive and non-preemptive jobs and arbitrary fractional release times and deadlines.
In terms of this problem, our result can be stated as follows: with respect to the set of job lengths, the only solvable case is the one in which
the non-preemptive jobs all have the same length.

\section{Acknowledgements}
The authors thank Sicco Verwer for useful discussions, Ji\v{r}i Sgall for valuable comments, Nodari Vakhania for valuable e-mail correspondence, and anonymous reviewers for valuable comments that improved the presentation.

\bibliographystyle{plainnat}
\bibliography{pub}

\end{document}